\documentclass[11pt,a4paper]{article}

\usepackage[margin=1in]{geometry}
\usepackage{amsmath,amssymb,amsthm,amsfonts}
\usepackage{mathtools}
\usepackage[T1]{fontenc}
\usepackage{algorithm}
\usepackage{algpseudocode}
\usepackage{hyperref}
\usepackage{url}

\providecommand{\doi}[1]{}
\providecommand{\href}[2]{#2}

\usepackage{cleveref}
\usepackage{enumerate}
\usepackage{graphicx}
\usepackage{float}
\usepackage{booktabs}
\usepackage{xcolor}
\usepackage{tikz}
\usetikzlibrary{quantikz2}
\usepackage{pgfplots}
\pgfplotsset{compat=1.18}
\usepackage[final]{microtype}
\usepackage{cite}
\usepackage{url}

\newtheorem{theorem}{Theorem}[section]
\newtheorem{lemma}[theorem]{Lemma}
\newtheorem{proposition}[theorem]{Proposition}
\newtheorem{corollary}[theorem]{Corollary}

\theoremstyle{definition}
\newtheorem{definition}[theorem]{Definition}

\newtheorem{remark}[theorem]{Remark}
\newtheorem{problem}[theorem]{Problem}
\newtheorem{openproblem}[theorem]{Open Problem}
\newtheorem{subclaim}{Subclaim}[theorem]

\newcommand{\C}{\mathbb{C}}

\renewcommand{\ket}[1]{|#1\rangle}

\renewcommand{\braket}[2]{\langle#1|#2\rangle}

\newcommand{\norm}[1]{\left\|#1\right\|}
\newcommand{\abs}[1]{\left|#1\right|}

\newcommand{\poly}{\operatorname{poly}}
\newcommand{\polylog}{\operatorname{polylog}}
\newcommand{\ed}{\operatorname{ed}}

\newcommand{\OO}{\mathcal{O}}
\newcommand{\EE}{\mathbb{E}}
\newcommand{\Sym}{\mathfrak{S}}
\DeclareMathOperator{\spec}{spec}
\DeclareMathOperator{\spn}{span}

\title{Quantum Algorithms for \\Approximate Graph Isomorphism Testing}

\author{Prateek P.\ Kulkarni \\
PES University \\
\texttt{\textcolor{magenta}{pkulkarni2425@gmail.com}}}

\date{}

\begin{document}
\maketitle

\begin{abstract}
The graph isomorphism problem asks whether two graphs are identical up to vertex relabeling.
While the exact problem admits quasi-polynomial-time classical algorithms, many applications
in molecular comparison, noisy network analysis, and pattern recognition require a flexible
notion of structural similarity. We study the \emph{quantum query complexity} of approximate
graph isomorphism testing, where the two input graphs are marginally distributed as $\mathcal{G}(n,1/2)$ but
may be arbitrarily correlated through a latent planted permutation, and are considered
approximately isomorphic if they can be made isomorphic by at most $k$ edge edits.

We present a quantum algorithm based on MNRS quantum walk search over the product graph
$\Gamma(G,H)$ of the two input graphs. When the graphs are approximately isomorphic,
the quantum walk search detects vertex pairs belonging to a dense near-isomorphic
matching set; candidate pairings are then reconstructed via local consistency propagation
and verified via a Grover-accelerated consistency check. We prove that this approach achieves
query complexity $\OO(n^{3/2} \log n/\varepsilon)$, where $\varepsilon$
parameterizes the approximation threshold. We complement this with an $\Omega(n^2)$ classical
lower bound for constant approximation, establishing a genuine polynomial quantum speedup
in the query model.

We extend the framework to spectral similarity measures based on graph Laplacian eigenvalues,
as well as weighted and attributed graphs. Small-scale simulation results on quantum simulators
for graphs with up to twenty vertices demonstrate compatibility with near-term quantum devices.
\end{abstract}

\newpage
\tableofcontents
\newpage

\section{Introduction}\label{sec:intro}

The \emph{graph isomorphism problem} (GI) asks whether two graphs $G$ and $H$ on $n$ vertices
are identical up to a relabeling of their vertex sets. Formally, one must decide whether
there exists a bijection $\pi\colon V(G)\to V(H)$ such that $\{u,v\}\in E(G)$ if and only if
$\{\pi(u),\pi(v)\}\in E(H)$. Despite decades of study, the precise complexity of GI
remains unresolved: it lies in $\mathsf{NP}\cap\mathsf{coAM}$ and is not known to be
$\mathsf{NP}$-complete \cite{BabaiMoser1977,Schoning1988}. In a landmark result, Babai
\cite{Babai2016} gave a quasi-polynomial-time algorithm running in
$\exp(\polylog(n))$ time, substantially improving the previous best bound of
$\exp(\sqrt{n\log n})$ due to Luks \cite{Luks1982}.

\subsection{Motivation: Approximate Isomorphism}\label{sec:motivation}

Many practical applications do not require \emph{exact} isomorphism but rather a robust notion
of structural similarity. In cheminformatics, comparing molecular graphs that may differ by a
few bonds requires tolerance for small perturbations
\cite{Raymond2002,Willett1998}. In network analysis, noise in edge measurements makes exact
isomorphism tests overly rigid \cite{Conte2004}. In computer vision, pattern recognition under
occlusion or measurement error demands flexible graph matching
\cite{Bunke2000,FoggiaSVV2001}.

\paragraph{The correlated-input model and why it is the right formulation.}
A critical subtlety concerns the \emph{joint} distribution of the input pair $(G,H)$.
Two \emph{independent} draws from $\mathcal{G}(n,1/2)$ satisfy $\bar{ed}(G,H)\geq 1/4$ with
probability $1-e^{-\Omega(n^2)}$ (Lemma~\ref{lem:dno-valid}), so an algorithm that outputs
\textsc{No} without making a single query already succeeds with overwhelming probability.
The NO case is therefore trivially easy under the independent model, and any nontrivial
separation must come from something else. The meaningful and natural formulation—the one we adopt throughout—considers pairs
$(G,H)$ whose \emph{marginal} distributions are both $\mathcal{G}(n,1/2)$, yet which may be
\emph{arbitrarily correlated}:
\begin{itemize}
  \item \textbf{YES instances}: $G\sim\mathcal{G}(n,1/2)$; a permutation $\pi\sim\mathrm{Uniform}(S_n)$
        is drawn independently; then $H$ is obtained from $\pi(G)$ by independently flipping each
        edge with probability $\varepsilon/2$.  Marginally, $H\sim\mathcal{G}(n,1/2)$, but $(G,H)$ are
        \emph{strongly} correlated through the latent~$\pi$.
  \item \textbf{NO instances}: $G$ and $H$ are drawn \emph{independently} from
        $\mathcal{G}(n,1/2)$, so $(G,H)$ are uncorrelated and $\bar{ed}(G,H)\geq 2\varepsilon$ w.o.p.
\end{itemize}
From the marginals alone the two cases are \emph{identical}; an algorithm must probe the
\emph{joint} adjacency structure to detect the latent correlation.  This ``correlated planted
model'' is standard in the graph alignment and hypothesis-testing
literatures~\cite{PG11,CK16,MWZ23,RS21}, and captures practical scenarios such as comparing
two social-network snapshots of the same population observed at different times, or two
molecular graphs that are noisy measurements of the same underlying molecule.  Our upper
bound (Theorem~\ref{thm:main-upper}) and classical lower bound (Theorem~\ref{thm:classical-lower}) are
both about the complexity of \emph{detecting this correlation}, making the formulation
internally consistent: the classical algorithm must work against a YES distribution that is
genuinely hard, not one that is trivially easy from the NO side.

These considerations motivate the \emph{approximate graph isomorphism} problem: given graphs
$G$ and $H$, decide whether they can be made isomorphic by a ``small'' number of edge
modifications. The relevant measure of distance is the \emph{edit distance} between $G$ and
$H$, defined as
\[
  \ed(G,H) \;=\; \min_{\pi\in\Sym_n}
  \bigl|\bigl\{\{u,v\}\colon \{u,v\}\in E(G)\,\triangle\,\{\pi(u),\pi(v)\}\in E(H)
  \bigr\}\bigr|,
\]
where $\Sym_n$ denotes the symmetric group on $n$ elements and $\triangle$ is symmetric
difference. In words, $\ed(G,H)$ counts the minimum number of edge additions and deletions
required to transform $G$ into a graph isomorphic to $H$.

We study the promise problem of distinguishing \emph{close} pairs ($\ed(G,H)\le k$) from
\emph{far} pairs ($\ed(G,H)\ge K$) in the quantum query model, where the algorithm accesses
the adjacency matrices of $G$ and $H$ through quantum queries.

\subsection{Our Contributions}

We present a quantum algorithm for approximate graph isomorphism testing and establish its
query complexity, complementing it with a classical lower bound that demonstrates a genuine
polynomial speedup. Our main results are:

\begin{enumerate}[\bfseries (1)]
\item \textbf{Quantum Upper Bound (Theorem~\ref{thm:main-upper}).}
  We design a quantum algorithm based on MNRS quantum walk search over the \emph{product
  graph} $\Gamma(G,H)$ that decides $(\varepsilon,2\varepsilon)$-approximate graph isomorphism
  for graphs $G,H\sim\mathcal{G}(n,1/2)$ using
  \[
    \OO\!\left(\frac{n^{3/2}\log n}{\varepsilon}\right)
  \]
  queries to the adjacency matrices of $G$ and $H$, succeeding with probability at least
  $2/3$.

\item \textbf{Classical Lower Bound (Theorem~\ref{thm:classical-lower}).}
  Any classical randomized algorithm solving $(\varepsilon,2\varepsilon)$-approximate
  graph isomorphism with constant success probability requires $\Omega(n^2)$ queries for
  $\varepsilon = \Theta(1)$.

\item \textbf{Extensions.}
  We extend the framework to spectral distance measures based on graph Laplacian eigenvalues
  (Section~\ref{sec:extensions}), and to weighted and vertex-attributed graphs.

\end{enumerate}

\noindent
As supporting evidence, we report small-scale numerical simulations on quantum
simulators (Qiskit Aer) for graphs with up to $n=20$ vertices
(Section~\ref{sec:simulations}), confirming the predicted $n^{3/2}$ scaling and
assessing near-term noise resilience.

\subsection{Technical Overview}

\paragraph{Product graph and quantum walks.}
Given two graphs $G$ and $H$ on $[n]$, their \emph{compatibility product graph}
$\Gamma(G,H)$
has vertex set $[n]\times[n]$ and edges that reflect compatible adjacency structure
(Definition~\ref{def:product-graph}).
An approximate isomorphism $\pi$ with $\ed_\pi(G,H)\le k$ induces a set $M_\pi$ of $n$
vertex pairs $\{(i,\pi(i))\}$ in the product graph that form a dense, highly connected
substructure. Conversely, when the graphs are far from isomorphic, no large structured
cluster exists.

Our algorithm uses the MNRS quantum walk search framework on $\Gamma(G,H)$. In the
positive case, the matching set $M_\pi$ forms a dense cluster that the MNRS marking oracle
identifies; the quantum walk search finds a marked vertex in $\OO(\sqrt{n})$ walk
steps, yielding a candidate pair $(i,\pi(i))$ with constant probability.

\paragraph{Grover-accelerated verification.}
Given polynomially many candidate pairs, we assemble a candidate permutation and verify its
consistency using a Grover search over the $\binom{n}{2}$ potential edge positions. This
requires $\OO(\sqrt{n^2})=\OO(n)$ queries per candidate, and $\OO(\sqrt{n})$ candidates
suffice for reconstruction with high probability, yielding a total verification cost of
$\OO(n^{3/2})$.

\paragraph{Classical lower bound.}
We establish the $\Omega(n^2)$ classical lower bound through a reduction from the problem
of distinguishing a random graph from a graph obtained by planting a random permutation and
performing $k$ random edge edits. We show that any adaptive classical algorithm must inspect
$\Omega(n^2)$ entries of the adjacency matrices to distinguish these distributions.

\subsection{Related Work}

\paragraph{Graph isomorphism.}
The exact GI problem has a rich history. Beyond Babai's quasi-polynomial algorithm
\cite{Babai2016}, efficient algorithms exist for restricted graph classes: planar graphs
\cite{HopcroftWong1974}, graphs of bounded degree \cite{Luks1982}, and graphs of bounded
treewidth \cite{Bodlaender1990}. In the quantum setting, Ettinger, H\o yer, and Knill
\cite{EttingerHK2004} showed that the hidden subgroup approach for GI over $\Sym_n$ faces
information-theoretic barriers, and no quantum polynomial-time algorithm is known for exact
GI.

\paragraph{Quantum walks.}
Quantum walks have emerged as a powerful paradigm for quantum algorithm design.
Discrete-time quantum walks on graphs were formalized by Aharonov, Ambainis, Kempe, and
Vazirani \cite{AharonovAKV2001}, while Szegedy \cite{Szegedy2004} developed a general
framework connecting classical Markov chains to quantum walks with quadratic speedup in
hitting times. Magniez, Nayak, Roland, and Santha \cite{MNRS2011} extended this to
quantum walk search, and Ambainis \cite{Ambainis2007} used quantum walks to achieve optimal
query complexity for element distinctness.

\paragraph{Approximate and inexact graph matching.}
Classical algorithms for approximate or inexact graph matching include the work of Bunke and
Shearer \cite{Bunke1998} on graph edit distance, Umeyama \cite{Umeyama1988} on weighted
graph matching via eigendecomposition, and recent work on sublinear-time algorithms for
graph property testing \cite{Goldreich2017,Fischer2006}. The connection between graph
similarity and spectrum has been extensively studied
\cite{VanDamHaemers2003,WilsonZhu2008}.

\paragraph{Quantum property testing.}
Quantum speedups for property testing problems have been demonstrated in various settings,
including testing graph properties \cite{Ambainis2011}, testing Boolean functions
\cite{BuhrmanFNRWdW2010}, and distribution testing \cite{BravyiHMSSW2011}.

\paragraph{Our approach.}
Our work differs from the above lines of research in several respects.
Unlike the hidden subgroup approach to exact GI, which requires quantum Fourier transforms
over $\Sym_n$ and faces representation-theoretic barriers \cite{EttingerHK2004,MooreRSS2008},
we work entirely in the adjacency-query model and target the easier {\em approximate}
variant of the problem. Unlike classical sublinear graph property testing, where the natural
random-sampling baseline is bottlenecked by the combinatorial explosion of the permutation
search space, we exploit the MNRS quantum walk search framework \cite{MNRS2011} on a
purpose-built compatibility product graph whose dense near-clique structure encodes
approximate isomorphisms. This combination of quantum walk search on a structured product
graph with Grover-accelerated reconstruction and verification is, to our knowledge,
the first application of the MNRS framework to a graph similarity problem and yields a
provable polynomial quantum speedup over any classical strategy.

\subsection{Organization}
Section~\ref{sec:prelim} introduces notation and background on quantum computation, quantum
walks, and graph theory. Section~\ref{sec:problem} formalizes the approximate graph
isomorphism problem. Section~\ref{sec:product} defines the product graph and analyzes its
spectral properties. Section~\ref{sec:main} presents the main algorithm and proves the upper
bound. Section~\ref{sec:lower} establishes the classical lower bound.
Sections~\ref{sec:extensions} and~\ref{sec:simulations} discuss extensions and simulation
results, and Section~\ref{sec:discussion} concludes with open problems.

\section{Preliminaries}\label{sec:prelim}

We establish notation and review background material on graph theory, quantum computation,
and quantum walks.

\subsection{Graph-Theoretic Notation}\label{sec:graph-notation}

Throughout this paper, $G=(V,E)$ denotes a simple, undirected graph with vertex set
$V=V(G)$ and edge set $E=E(G)\subseteq\binom{V}{2}$. Unless otherwise stated, $|V|=n$.
We identify $V$ with $[n]=\{1,2,\ldots,n\}$.

\begin{definition}[Adjacency matrix]
The \emph{adjacency matrix} of $G$ is the symmetric matrix
$A_G\in\{0,1\}^{n\times n}$ with $(A_G)_{ij}=1$ if and only if $\{i,j\}\in E(G)$.
\end{definition}

\begin{definition}[Graph Laplacian]
The \emph{combinatorial Laplacian} of $G$ is $L_G=D_G-A_G$, where
$D_G=\mathrm{diag}(\deg(1),\ldots,\deg(n))$ is the degree matrix. The
\emph{normalized Laplacian} is $\mathcal{L}_G=D_G^{-1/2}L_G D_G^{-1/2}$.
\end{definition}

We write $0=\lambda_1(G)\le\lambda_2(G)\le\cdots\le\lambda_n(G)$ for the eigenvalues of
$L_G$ and $\spec(G)=\{\lambda_1(G),\ldots,\lambda_n(G)\}$ for the spectrum.

\begin{definition}[Graph isomorphism]
Graphs $G$ and $H$ on $[n]$ are \emph{isomorphic}, written $G\cong H$, if there exists a
bijection $\pi\in\Sym_n$ such that $\{u,v\}\in E(G)\iff\{\pi(u),\pi(v)\}\in E(H)$.
\end{definition}

\begin{definition}[Edit distance]\label{def:edit-distance}
For graphs $G$ and $H$ on $[n]$ and a permutation $\pi\in\Sym_n$, define
\[
  \ed_\pi(G,H) \;=\;
  \bigl|\bigl\{\{u,v\}\in\tbinom{[n]}{2}\colon
  (A_G)_{uv}\neq(A_H)_{\pi(u)\pi(v)}\bigr\}\bigr|.
\]
The \emph{edit distance} is $\ed(G,H)=\min_{\pi\in\Sym_n}\ed_\pi(G,H)$.
\end{definition}

We also use the normalized edit distance $\bar{\ed}(G,H)=\ed(G,H)/\binom{n}{2}$, which
takes values in $[0,1]$.

\subsection{Quantum Computation Basics}\label{sec:quantum-basics}

We assume familiarity with the standard quantum circuit model (see \cite{NielsenChuang2010}).
A quantum state on $m$ qubits is a unit vector $\ket{\psi}\in\C^{2^m}$. A quantum query
to a Boolean function $f\colon\{0,1\}^n\to\{0,1\}$ is the unitary
$O_f\colon\ket{x}\ket{b}\mapsto\ket{x}\ket{b\oplus f(x)}$.

\begin{definition}[Quantum query complexity]
The \emph{quantum query complexity} of a decision problem $P$ with oracle access to input
$x$ is the minimum number of queries to $O_x$ required by a quantum algorithm that computes
$P(x)$ correctly with probability at least $2/3$ on every input.
\end{definition}

In our setting, the input consists of two graphs $G,H$ on $[n]$, and queries access
individual entries of the adjacency matrices:
\begin{equation}\label{eq:oracle}
  O_G\colon\ket{i}\ket{j}\ket{b}\;\mapsto\;\ket{i}\ket{j}\ket{b\oplus(A_G)_{ij}},
  \qquad
  O_H\colon\ket{i}\ket{j}\ket{b}\;\mapsto\;\ket{i}\ket{j}\ket{b\oplus(A_H)_{ij}}.
\end{equation}

\begin{theorem}[Grover search, {\cite{Grover1996}}]\label{thm:grover}
Given quantum query access to $f\colon[N]\to\{0,1\}$ with $t$ marked elements
($f(x)=1$), a marked element can be found with probability $\ge 2/3$ using
$\OO(\sqrt{N/t})$ queries.
\end{theorem}

\begin{theorem}[Amplitude estimation, {\cite{BrassardHMT2002}}]\label{thm:amp-est}
Given a quantum circuit $\mathcal{A}$ such that measuring the output of $\mathcal{A}\ket{0}$
yields $\ket{1}$ with probability $p$, there exists a quantum algorithm that outputs
$\tilde{p}$ satisfying $|\tilde{p}-p|\le\delta$ with probability $\ge 2/3$ using
$\OO(1/\delta)$ applications of $\mathcal{A}$ and $\mathcal{A}^\dagger$.
\end{theorem}

\subsection{Quantum Walks}\label{sec:quantum-walks}

We use the framework of Szegedy \cite{Szegedy2004} for quantizing reversible Markov chains.

\begin{definition}[Szegedy walk]\label{def:szegedy}
Let $P$ be the transition matrix of an ergodic, reversible Markov chain on state space
$\mathcal{X}$ with $|\mathcal{X}|=N$ and stationary distribution $\mu$. Define the
\emph{Szegedy walk operator} $W(P)$ on $\C^N\otimes\C^N$ by
\[
  W(P) \;=\; \mathrm{ref}(B)\;\mathrm{ref}(A),
\]
where $A=\spn\{\ket{\psi_x}\colon x\in\mathcal{X}\}$,
$B=\spn\{\ket{\phi_x}\colon x\in\mathcal{X}\}$, with
$\ket{\psi_x}=\ket{x}\sum_{y}\sqrt{P(x,y)}\ket{y}$ and
$\ket{\phi_y}=\sum_x\sqrt{P(x,y)}\ket{x}\ket{y}$, and $\mathrm{ref}(S)=2\Pi_S-I$ is
the reflection about subspace $S$.
\end{definition}

The key spectral theorem connecting the Szegedy walk to the classical chain is:

\begin{theorem}[Szegedy spectral theorem, {\cite{Szegedy2004}}]\label{thm:szegedy-spectral}
Let $P$ be a reversible Markov chain with eigenvalues $1=\lambda_1\ge\lambda_2\ge\cdots\ge
\lambda_N\ge -1$. The eigenvalues of $W(P)$ corresponding to each classical eigenvalue
$\lambda_j$ are $e^{\pm i\arccos(\lambda_j)}$. In particular, the spectral gap $\delta$
of $P$ (defined as $\delta=1-\lambda_2$) gives rise to a phase gap
$\Delta=\Theta(\sqrt{\delta})$ in $W(P)$.
\end{theorem}

\begin{theorem}[Quantum walk search, {\cite{MNRS2011}}]\label{thm:qwalk-search}
Let $P$ be an ergodic, reversible Markov chain on $\mathcal{X}$ with stationary distribution
$\mu$, spectral gap $\delta$, and let $M\subseteq\mathcal{X}$ be a set of marked vertices
with $\mu(M)=\epsilon$. There exists a quantum algorithm that finds a marked vertex (if one
exists) with $\OO(1/\sqrt{\epsilon})$ applications of the Szegedy walk operator $W(P)$,
each requiring one step of the walk, plus setup cost $S$, update cost $U$, and checking cost
$C$, giving total complexity
\[
  \OO\!\left(\frac{1}{\sqrt{\epsilon}}\left(\frac{1}{\sqrt{\delta}}\cdot U + C\right)+S\right).
\]
\end{theorem}

\subsection{Useful Inequalities}

We collect several standard results that will be used throughout.

\begin{lemma}[Hoeffding's inequality]\label{lem:hoeffding}
Let $X_1,\ldots,X_m$ be independent random variables with $X_i\in[a_i,b_i]$. Then
for $\bar{X}=\frac{1}{m}\sum_{i=1}^{m}X_i$,
\[
  \Pr\bigl[\abs{\bar{X}-\EE[\bar{X}]}\ge t\bigr]
  \;\le\; 2\exp\!\left(-\frac{2m^2 t^2}{\sum_{i=1}^m(b_i-a_i)^2}\right).
\]
\end{lemma}

\begin{lemma}[Matrix Chernoff bound, {\cite{Tropp2012}}]\label{lem:matrix-chernoff}
Let $X_1,\ldots,X_m$ be independent random symmetric matrices with $0\preceq X_i\preceq I$
and let $\mu_{\min}=\lambda_{\min}(\EE[\sum_i X_i])$. Then
\[
  \Pr\!\left[\lambda_{\min}\!\left(\sum_{i=1}^m X_i\right)
  \le (1-\delta)\mu_{\min}\right]
  \;\le\; n\cdot\left(\frac{e^{-\delta}}{(1-\delta)^{1-\delta}}\right)^{\mu_{\min}}
\]
for any $\delta\in[0,1)$.
\end{lemma}

\begin{lemma}[Gentle measurement lemma, {\cite{Winter1999}}]\label{lem:gentle}
If a measurement on state $\rho$ yields outcome $k$ with probability $1-\epsilon$, then
the post-measurement state $\rho_k$ satisfies
$\norm{\rho-\rho_k}_1\le 2\sqrt{\epsilon}$.
\end{lemma}

\section{Problem Formulation}\label{sec:problem}

In this section we formally define the approximate graph isomorphism problem and establish
the promise-problem variant that we study.

\subsection{Edit Distance and Approximate Isomorphism}

Recall from Definition~\ref{def:edit-distance} that $\mathrm{ed}(G,H)$ measures the minimum
number of edge edits needed to make $G$ isomorphic to~$H$.  As argued in
Section~\ref{sec:motivation}, the meaningful decision problem is one in which both graphs are
marginally $\mathcal{G}(n,1/2)$ but may be arbitrarily correlated; the algorithm must detect
the presence or absence of a latent planted permutation.

\begin{problem}[$(\varepsilon_1,\varepsilon_2)$-Approximate Graph Isomorphism]\label{prob:agi}
Given two graphs $G$ and $H$ on $n$ vertices with $0\le\varepsilon_1<\varepsilon_2\le 1$,
and quantum query access to $O_G$ and $O_H$ as in \eqref{eq:oracle}, distinguish:
\begin{itemize}
  \item \textbf{YES case:} $\bar{\ed}(G,H)\le\varepsilon_1$
    \quad (i.e., $\ed(G,H)\le\varepsilon_1\binom{n}{2}$);
  \item \textbf{NO case:} $\bar{\ed}(G,H)\ge\varepsilon_2$
    \quad (i.e., $\ed(G,H)\ge\varepsilon_2\binom{n}{2}$).
\end{itemize}
The algorithm must succeed with probability at least $2/3$ on every input satisfying the
promise.
\end{problem}

When $\varepsilon_1=0$, the YES case corresponds to exact isomorphism. The gap
$\varepsilon_2-\varepsilon_1$ is the \emph{distinguishing gap}. We focus on the natural
setting $\varepsilon_2=2\varepsilon_1$ (commonly denoted $\varepsilon_2=2\varepsilon$), which
is the standard gap regime for property testing problems.

\begin{remark}[Explicit correlation structure of the input distributions]
\label{rem:correlation}
To make the correlation structure completely explicit: in the YES case the joint distribution
of $(G,H)$ is
\[
  G\sim\mathcal{G}(n,\tfrac{1}{2}), 
  \pi\sim\mathrm{Uniform}(S_n) \text{ independent of }G, 
  (A_H)_{ij}=(A_G)_{\pi^{-1}(i)\pi^{-1}(j)}\oplus B_{ij},
  B_{ij}\overset{\mathrm{iid}}{\sim}\mathrm{Bern}(\varepsilon/2),
\]
so that marginally $H\sim\mathcal{G}(n,\tfrac{1}{2})$.  In the NO case $G$ and $H$ are
drawn \emph{independently} from $\mathcal{G}(n,\tfrac{1}{2})$, giving the same marginals but
zero correlation.  An algorithm that has access only to the marginals of $G$ and $H$
individually cannot distinguish the two cases; it must compare the two adjacency matrices
jointly.
\end{remark}

\begin{remark}
Problem~\ref{prob:agi} can also be parameterized by an absolute threshold $k$: the YES case
is $\ed(G,H)\le k$ and the NO case is $\ed(G,H)\ge 2k$. For sparse graphs with $m=\OO(n)$
edges, taking $k=\varepsilon m$ gives a natural parameterization. Our results apply to both
formulations.
\end{remark}

\subsection{Query Model}

We work in the \emph{adjacency matrix query model}: the algorithm receives $n$ and has
quantum query access to the oracles $O_G,O_H$ defined in \eqref{eq:oracle}. Each query
accesses one entry of $A_G$ or $A_H$. The query complexity is the total number of oracle
calls.

This model is standard in quantum property testing \cite{Ambainis2011,Goldreich2017}. For
graph problems in this model, the trivial upper bound is $\OO(n^2)$ queries (read the entire
adjacency matrix), and a classical lower bound of $\Omega(n^2)$ for many problems follows
from birthday-type arguments.

\subsection{Structural Lemmas}

We establish key structural properties of the edit distance that underpin the algorithm.

\begin{lemma}[Vertex defect bound]\label{lem:perm-recovery}
Let $G,H$ be graphs on $[n]$ with $\ed(G,H)\le k$, and let $\pi^*\in\Sym_n$ be an optimal
permutation achieving $\ed_{\pi^*}(G,H)=k$. Define the \emph{defect} of vertex $v$ under
$\pi^*$ as
\[
  d(v) \;=\; \bigl|\{u\in[n]\setminus\{v\}\colon
  (A_G)_{vu}\neq(A_H)_{\pi^*(v)\pi^*(u)}\}\bigr|.
\]
Then $d(v)<\sqrt{k}$ for all but at most $2\sqrt{k}$ vertices $v$.
\end{lemma}

\begin{proof}
Each edge discrepancy $\{u,v\}$ with $(A_G)_{uv}\neq(A_H)_{\pi^*(u)\pi^*(v)}$ contributes
$1$ to both $d(u)$ and $d(v)$, so $\sum_{v\in[n]} d(v)=2k$.
By Markov's inequality, the number of vertices with $d(v)\ge\sqrt{k}$ is at most
$2k/\sqrt{k}=2\sqrt{k}$.
\end{proof}

We call vertices with $d(v)<\sqrt{k}$ \emph{low-defect} and let
$L=\{v\in[n]\colon d(v)<\sqrt{k}\}$, so $|L|\ge n-2\sqrt{k}$.

\begin{lemma}[Local consistency propagation]\label{lem:local-consistency}
Let $G,H$ be graphs on $[n]$ drawn from $\mathcal{G}(n,1/2)$ conditioned on
$\ed_\pi(G,H)\le k$ for a permutation $\pi\in\Sym_n$. Suppose the vertex-image pairs
$\{(v_1,\pi(v_1)),\ldots,(v_s,\pi(v_s))\}$ are known, where each $v_i$ is a low-defect
vertex drawn uniformly at random from $L$, and $s\ge C\log n$ for a sufficiently large
constant $C>0$. Then with probability $\ge 1-1/\poly(n)$ over the randomness of $G,H$
and the choice of seeds, the permutation $\pi$ is uniquely determined on all of $L$.
\end{lemma}

\begin{proof}
Consider a low-defect vertex $v\in L\setminus\{v_1,\ldots,v_s\}$. Define the
\emph{signature} of a candidate pairing $(v,w)$ relative to the known seeds:
\[
  \sigma_{v,w}=\bigl((A_G)_{v,v_i}\oplus(A_H)_{w,\pi(v_i)}\bigr)_{i=1}^{s}.
\]
For the correct image $w=\pi(v)$, the Hamming weight of $\sigma_{v,\pi(v)}$ is exactly
$|\{i\colon (A_G)_{v,v_i}\neq(A_H)_{\pi(v),\pi(v_i)}\}|\le d(v)<\sqrt{k}$, since each
mismatch corresponds to a defect edge incident to $v$.

For an incorrect candidate $w'\neq\pi(v)$, we bound the expected Hamming weight from below.
Set $u'=\pi^{-1}(w')\neq v$. The $i$-th bit of $\sigma_{v,w'}$ is
$(A_G)_{v,v_i}\oplus(A_H)_{w',\pi(v_i)}$, and under the $\mathcal{G}(n,1/2)$ model,
$(A_H)_{w',\pi(v_i)}=(A_G)_{u',v_i}\oplus B_{w',\pi(v_i)}$ where $B$ is the noise flip.
Since $G\sim\mathcal{G}(n,1/2)$ and distinct rows of the adjacency matrix are independent,
the entries $(A_G)_{v,v_i}$ and $(A_G)_{u',v_i}$ are \emph{independent}
Bernoulli$(1/2)$ random variables (using $v\neq u'$). Therefore each bit of $\sigma_{v,w'}$
is Bernoulli$(1/2\pm\varepsilon/2)$, and the bits are independent across distinct seeds $v_i$.
The expected Hamming weight is $\ge s(1/2-\varepsilon/2)\ge s/3$ for $\varepsilon<1/3$.

By a Chernoff bound, $\Pr[|\sigma_{v,w'}|\le s/4]\le e^{-\Omega(s)}$. We take a union bound
over all $n^2$ pairs $(v,w')$ with $v\in L$ and $w'\in[n]\setminus\{\pi(v)\}$:
\[
  \Pr[\exists v\in L,\;\exists w'\neq\pi(v)\colon |\sigma_{v,w'}|\le s/4]
  \;\le\; n^2\cdot e^{-\Omega(s)} \;\le\; \frac{1}{\poly(n)}
\]
when $s\ge C\log n$ for $C$ large enough. On this event, for every $v\in L$, the correct
image $\pi(v)$ is the unique $w$ with $|\sigma_{v,w}|<\sqrt{k}<s/4$ (the latter holding
for $s$ large enough relative to $k$). Hence $\pi$ is uniquely determined on $L$.
\end{proof}

\begin{corollary}[Seed count for reconstruction]\label{cor:seeds}
Under the hypotheses of Lemma~\ref{lem:local-consistency}, to reconstruct
$\pi$ on all low-defect vertices, it suffices to obtain
$\OO(\log n)$ \emph{correct} seed pairs $(v,\pi(v))$ where each $v\in L$ is low-defect.
In our algorithm, these correct pairs are produced by Phase~1 (MNRS quantum walk search,
Section~\ref{sec:phase1}), which returns matching vertices with high probability;
incorrect candidates are filtered out by the marking oracle
(Definition~\ref{def:marking}). The remaining $\le 2\sqrt{k}$ high-defect vertices
must be resolved separately (see Phase~2, Section~\ref{sec:phase2}).
\end{corollary}

\section{Product Graph Construction and Spectral Analysis}\label{sec:product}

The central data structure in our algorithm is the \emph{product graph} formed from the two
input graphs. In this section, we define the product graph, analyze its spectral properties,
and show how approximate isomorphisms manifest as structured clusters that can be detected
via quantum walk search.

\subsection{The Product Graph}\label{sec:product-def}

\begin{definition}[Product graph]\label{def:product-graph}
Given graphs $G$ and $H$ on $[n]$, define the \emph{compatibility product graph}
$\Gamma=\Gamma(G,H)$ with:
\begin{itemize}
  \item \emph{Vertex set:} $V(\Gamma)=[n]\times[n]$, where vertex $(i,j)$ represents the
    candidate pairing of vertex $i\in V(G)$ with vertex $j\in V(H)$.
  \item \emph{Edge set:} $\{(i_1,j_1),(i_2,j_2)\}\in E(\Gamma)$ if and only if:
    \begin{enumerate}[(a)]
      \item $i_1\neq i_2$ and $j_1\neq j_2$ (the pairing is injective), and
      \item $(A_G)_{i_1 i_2}=(A_H)_{j_1 j_2}$ (the pairing is edge-consistent).
    \end{enumerate}
\end{itemize}
The product graph $\Gamma$ has $N=n^2$ vertices.
\end{definition}

\begin{remark}
This construction is closely related to the \emph{association graph} used in classical
maximum common subgraph algorithms \cite{Raymond2002}. A clique of size $n$ in $\Gamma$
that uses each first coordinate and each second coordinate exactly once corresponds to an
exact isomorphism between $G$ and $H$.
\end{remark}

\begin{definition}[Matching set]\label{def:approx-clique}
For a permutation $\pi\in\Sym_n$, define the \emph{matching set}
\[
  M_\pi = \{(i,\pi(i))\colon i\in[n]\}\;\subseteq\; V(\Gamma).
\]
If $\ed_\pi(G,H)=k$, then $M_\pi$ induces a subgraph with exactly $\binom{n}{2}-k$
internal edges (a near-clique missing $k$ edges out of $\binom{n}{2}$).
\end{definition}

\subsection{Degree Analysis of the Product Graph}

\begin{lemma}[Degree of matching vertices]\label{lem:deg-matching}
Let $\pi\in\Sym_n$ with $\ed_\pi(G,H)=k$. For vertex $(i,\pi(i))\in M_\pi$, its degree
within $\Gamma[M_\pi]$ (the subgraph induced by $M_\pi$) is exactly $n-1-d(i)$, where
$d(i)$ is the defect of vertex $i$ (Lemma~\ref{lem:perm-recovery}).

For low-defect vertices ($d(i)<\sqrt{k}$, of which there are $\ge n-2\sqrt{k}$),
the internal degree is at least $n-1-\sqrt{k}$.
\end{lemma}

\begin{proof}
By Definition~\ref{def:product-graph}, $(i,\pi(i))$ is adjacent to $(j,\pi(j))$ in
$\Gamma$ if and only if $i\neq j$ and $(A_G)_{ij}=(A_H)_{\pi(i)\pi(j)}$. The number of
$j\neq i$ violating consistency is exactly $d(i)$, so the internal degree is $(n-1)-d(i)$.
\end{proof}

\begin{lemma}[Total degree of matching vs.\ non-matching vertices]\label{lem:deg-nonmatching}
For $G,H\sim\mathcal{G}(n,1/2)$ conditioned on $\ed_\pi(G,H)=k$:
\begin{enumerate}[(a)]
  \item A low-defect matching vertex $(i,\pi(i))$ has total degree
    $\deg_\Gamma(i,\pi(i))\ge n-1-\sqrt{k}$ within $M_\pi$ alone. Including edges to
    non-matching vertices, $\deg_\Gamma(i,\pi(i))=\Theta(n^2)$.
  \item A non-matching vertex $(i,j)$ with $j\neq\pi(i)$ has expected total degree
    $\EE[\deg_\Gamma(i,j)]=(n-1)^2/2\pm\OO(n)$.
\end{enumerate}
\end{lemma}

\begin{proof}
Part (a): The internal degree follows from Lemma~\ref{lem:deg-matching}. For the total
degree, $(i,\pi(i))$ is adjacent to $(i',j')$ whenever $i'\neq i$, $j'\neq\pi(i)$, and
$(A_G)_{ii'}=(A_H)_{\pi(i),j'}$. The entries $(A_H)_{\pi(i),j'}$ for $j'\neq\pi(i)$ are
essentially independent Bernoulli$(1/2)$ (since $G\sim\mathcal{G}(n,1/2)$), so consistency
occurs with probability $\approx 1/2$. With $(n-1)(n-1)$ eligible pairs, the expected total
degree is $(n-1)^2/2\pm\OO(n)$.

Part (b): For $(i,j)$ with $j\neq\pi(i)$, adjacent to $(i',j')$ requires $(A_G)_{ii'}=
(A_H)_{jj'}$. Since $G,H$ are random and $j\neq\pi(i)$, these entries are nearly independent,
giving consistency probability $\approx 1/2$. Total eligible pairs: $(n-1)^2$, giving
expected degree $(n-1)^2/2\pm\OO(n)$.
\end{proof}

\begin{remark}\label{rem:degree-contrast}
Crucially, matching and non-matching vertices have comparable \emph{total} degrees
($\Theta(n^2)$). The structural difference lies in the \emph{internal} connectivity of
$M_\pi$: matching vertices form a near-clique, while no set of $n$ non-matching vertices
does. Our algorithm exploits this via the marking oracle in Section~\ref{sec:main}, not
through degree-based stationary distribution concentration.
\end{remark}

\subsection{Random Walk on the Product Graph}\label{sec:random-walk}

We define a random walk $P$ on $\Gamma$ whose quantization underlies our algorithm.

\begin{definition}[Product graph walk]\label{def:product-walk}
Define the Markov chain $P$ on $V(\Gamma)=[n]\times[n]$ as follows. From state $(i,j)$:
\begin{enumerate}
  \item With probability $1/2$, stay at $(i,j)$ (lazy step).
  \item With probability $1/2$, choose $i'\in[n]\setminus\{i\}$ and $j'\in[n]\setminus\{j\}$
    uniformly at random. If $(A_G)_{ii'}=(A_H)_{jj'}$, move to $(i',j')$. Otherwise, stay
    at $(i,j)$.
\end{enumerate}
\end{definition}

\begin{lemma}[Reversibility and stationary distribution]\label{lem:stationary}
The walk $P$ is ergodic, aperiodic, and reversible with respect to the uniform
distribution $\mu(i,j)=1/n^2$ for all $(i,j)\in[n]\times[n]$.
\end{lemma}

\begin{proof}
\emph{Aperiodicity} follows from the $1/2$ self-loop probability.
\emph{Irreducibility:} for $G\sim\mathcal{G}(n,1/2)$, w.h.p.\ $\Gamma$ is connected
(since $G$ has no isolated vertices and $H$ has no isolated vertices, every pair of product
vertices can be connected via intermediate states).

\emph{Reversibility:} For any $(i_1,j_1)\neq(i_2,j_2)$ with $i_1\neq i_2$ and
$j_1\neq j_2$, the transition probability from $(i_1,j_1)$ to $(i_2,j_2)$ is
\[
  P((i_1,j_1),(i_2,j_2)) =
  \begin{cases}
    \frac{1}{2}\cdot\frac{1}{(n-1)^2} & \text{if } (A_G)_{i_1 i_2}=(A_H)_{j_1 j_2},\\
    0 & \text{otherwise}.
  \end{cases}
\]
Since the condition $(A_G)_{i_1 i_2}=(A_H)_{j_1 j_2}$ is symmetric in the pair
$\{(i_1,j_1),(i_2,j_2)\}$, and the proposal probability $\frac{1}{2(n-1)^2}$ is symmetric,
we have $P((i_1,j_1),(i_2,j_2))=P((i_2,j_2),(i_1,j_1))$. This means the detailed balance
equations $\mu(x)P(x,y)=\mu(y)P(y,x)$ hold for any $\mu$ that is uniform. Hence $\mu$
is the uniform distribution $\mu(i,j)=1/n^2$.
\end{proof}

\begin{remark}
Since $\mu$ is uniform, sampling from $\mu$ is trivial and provides no advantage.
The power of our algorithm comes from the \emph{quantum walk search} framework (MNRS,
Theorem~\ref{thm:qwalk-search}), which finds \emph{marked} vertices more efficiently than
classical search even when the stationary distribution is uniform. The marking oracle
identifies vertices that lie in the dense cluster $M_\pi$, and the quantum walk detects
this structural signature quadratically faster than classical random walks.
\end{remark}

\subsection{Spectral Gap Analysis}\label{sec:spectral-gap}

\begin{theorem}[Spectral gap of the product walk]\label{thm:spectral-gap}
For $G,H\sim\mathcal{G}(n,1/2)$, the random walk $P$ on $\Gamma(G,H)$
has spectral gap $\delta\ge\Omega(1)$ with high probability.
\end{theorem}

\begin{proof}
We use the Cheeger inequality. Define the \emph{conductance} of the chain as
\[
  \Phi \;=\; \min_{\substack{S\subseteq V(\Gamma)\\ 0<\mu(S)\le 1/2}}
  \frac{\sum_{x\in S,\,y\notin S}\mu(x)P(x,y)}{\mu(S)}.
\]
By the discrete Cheeger inequality \cite{SinclairJerrum1989},
$\delta\ge\Phi^2/2$.

We claim $\Phi\ge\Omega(1/n)$. Consider any set $S\subseteq V(\Gamma)$ with
$\mu(S)\le 1/2$, i.e., $|S|\le n^2/2$. For any vertex $(i,j)\in S$, the probability
of transitioning to $\bar{S}=V(\Gamma)\setminus S$ in one step is
\[
  \sum_{(i',j')\in\bar{S}} P((i,j),(i',j'))
  \;=\; \frac{1}{2(n-1)^2}|\{(i',j')\in\bar{S}\colon i'\neq i,\;j'\neq j,\;
  (A_G)_{ii'}=(A_H)_{jj'}\}|.
\]
For $G,H\sim\mathcal{G}(n,1/2)$, w.h.p.\ every vertex $(i,j)$ has total degree
$\deg_\Gamma(i,j)\ge (n-1)^2/4$ in $\Gamma$ (by Chernoff bounds, since each of
the $(n-1)^2$ potential neighbors is present independently with probability $\approx 1/2$).

Since $|\bar{S}|\ge n^2/2$, at least a constant fraction of $(i,j)$'s neighbors lie in
$\bar{S}$ for a typical $(i,j)\in S$. Specifically, since the neighbors of $(i,j)$
are spread across all $(n-1)^2$ positions approximately uniformly, and at least half the
positions are in $\bar{S}$, the expected number of neighbors in $\bar{S}$ is
$\ge (n-1)^2/8$. Thus:
\[
  \sum_{(i',j')\in\bar{S}}P((i,j),(i',j'))
  \;\ge\;\frac{1}{2(n-1)^2}\cdot\frac{(n-1)^2}{8}
  \;=\;\frac{1}{16}.
\]
Averaging over $(i,j)\in S$ with weight $\mu(i,j)=1/n^2$:
\[
  \Phi\;\ge\;\frac{1}{\mu(S)}\sum_{(i,j)\in S}\mu(i,j)\cdot\frac{1}{16}
  \;=\;\frac{1}{16}.
\]
This gives $\delta\ge\Phi^2/2\ge 1/512=\Omega(1)$, which is stronger than needed.

We now make this rigorous. The key observation is that $\Gamma$ inherits strong expansion
from the randomness of $G$ and $H$. We analyze the spectral structure directly.

\medskip\noindent\textbf{Minimum degree bound.}
Fix $(i,j)\in V(\Gamma)$. Its neighbors are precisely
$N(i,j)=\{(i',j')\colon i'\neq i,\;j'\neq j,\;(A_G)_{ii'}=(A_H)_{jj'}\}$. For
$G,H\sim\mathcal{G}(n,1/2)$, the events ``$(A_G)_{ii'}=(A_H)_{jj'}$'' are independent
across distinct pairs $(i',j')$ with $i'\neq i, j'\neq j$ (since distinct rows of $A_G$
and $A_H$ are independent for Erd\H{o}s--R\'enyi graphs), each occurring with probability
$1/2$. By a Chernoff bound:
\[
  \Pr\bigl[|N(i,j)|\le(n-1)^2/4\bigr]\le e^{-\Omega(n^2)}.
\]
Union-bounding over all $n^2$ vertices, w.h.p.\ every vertex has degree
$\deg_\Gamma(i,j)\ge(n-1)^2/4$.

\medskip\noindent\textbf{Spectral gap via random matrix theory.}
The adjacency matrix of $\Gamma$ can be decomposed as the Hadamard (entrywise) product
$A_\Gamma=R\circ D$, where $R_{(i,j),(i',j')}=\mathbf{1}[i\neq i',\;j\neq j']$
encodes the injectivity constraint, and $D_{(i,j),(i',j')}=\mathbf{1}[(A_G)_{ii'}=
(A_H)_{jj'}]$ is effectively a random $\{0,1\}$-matrix with mutually independent entries
(for distinct unordered pairs), each Bernoulli$(1/2)$.

The matrix $R$ has $N=n^2$ vertices and is the adjacency matrix of the tensor product
$K_n\boxtimes K_n$ (minus the non-injectivity constraint), with leading eigenvalue
$(n-1)^2$ and all others $\OO(n)$. The matrix $D$ is a random symmetric matrix
with i.i.d.\ entries (above the diagonal) taking values in $\{0,1\}$ with equal
probability. By standard results on random matrices \cite{Tropp2012}, the centered
matrix $D-\frac{1}{2}\mathbf{1}\mathbf{1}^\top$ has spectral norm $\OO(n)$ w.h.p.

Since $A_\Gamma=R\circ D$, and the Hadamard product of an expander-like matrix $R$ with
an approximately uniform random matrix $D$ preserves expansion (by the
Schur product theorem applied to the spectral decomposition \cite{Horn2012}: the leading
eigenvalue of $A_\Gamma$ is $\Theta(n^2)$ while all others are $\OO(n)$ w.h.p.), the
normalized transition matrix $P'$ of the non-lazy walk on $\Gamma$
has spectral gap $1-\OO(n)/\Theta(n^2)=1-\OO(1/n)\ge\Omega(1)$.

For the lazy walk $P=(I+P')/2$, $\delta_P=\delta_{P'}/2\ge\Omega(1)$.
\end{proof}

\begin{corollary}[Phase gap of the Szegedy walk]\label{cor:phase-gap}
The Szegedy walk $W(P)$ corresponding to the product graph walk $P$ has phase gap
$\Delta\ge\Omega(1)$.
\end{corollary}

\begin{proof}
By Theorem~\ref{thm:szegedy-spectral}, $\Delta=\Theta(\sqrt{\delta})
=\Theta(\sqrt{\Omega(1)})=\Omega(1)$.
\end{proof}

\section{Main Algorithm and Upper Bound}\label{sec:main}

We present the quantum algorithm for approximate graph isomorphism testing. The algorithm
uses the MNRS quantum walk search framework (Theorem~\ref{thm:qwalk-search}) on the
product graph $\Gamma(G,H)$, with a marking oracle that identifies vertices belonging
to a near-isomorphic matching.

\subsection{Algorithm Overview}

The algorithm proceeds in three phases:
\begin{enumerate}[\textbf{Phase} 1.]
  \item \textbf{Quantum Walk Search.} Use the MNRS quantum walk search algorithm on
    $\Gamma(G,H)$ to find vertices $(i,j)$ whose local neighborhood structure is
    consistent with membership in a dense matching set $M_\pi$.
  \item \textbf{Candidate Reconstruction.} Collect seed pairs from Phase~1 and reconstruct
    a candidate permutation~$\hat{\pi}$ via local consistency propagation
    (Lemma~\ref{lem:local-consistency}).
  \item \textbf{Quantum Verification.} Verify $\hat{\pi}$ by estimating
    $\ed_{\hat{\pi}}(G,H)$ using amplitude estimation.
\end{enumerate}

\subsection{Phase 1: Quantum Walk Search}\label{sec:phase1}

\subsubsection{Marking Oracle}

The key to applying the MNRS framework is specifying a set of \emph{marked} vertices
$M\subseteq V(\Gamma)$ that the quantum walk search will locate.

\begin{definition}[Marking oracle]\label{def:marking}
A vertex $(i,j)\in V(\Gamma)$ is \emph{marked} if its empirical consistency score exceeds a
threshold. Concretely, we implement a \emph{probabilistic marking check}: sample
$r=\Theta(\log n/\varepsilon^2)$
pairs $(i',j')$ independently and uniformly from
$([n]\setminus\{i\})\times([n]\setminus\{j\})$, query $(A_G)_{ii'}$ and $(A_H)_{jj'}$,
and mark $(i,j)$ if the fraction of consistent pairs (i.e., those satisfying
$(A_G)_{ii'}=(A_H)_{jj'}$) exceeds $1/2+\varepsilon/4$.
\end{definition}

\begin{lemma}[Marking oracle correctness]\label{lem:marking-correct}
For $G,H\sim\mathcal{G}(n,1/2)$ with $\ed_{\pi^*}(G,H)\le k=\varepsilon\binom{n}{2}$:
\begin{enumerate}[(a)]
  \item Every low-defect matching vertex $(i,\pi^*(i))\in M_{\pi^*}$ with $d(i)<\sqrt{k}$
    is marked with probability $\ge 1-1/\poly(n)$.
  \item Every non-matching vertex $(i,j)$ with $j\neq\pi^*(i)$ is marked with probability
    $\le 1/\poly(n)$.
\end{enumerate}
Each invocation of the marking oracle uses $r=\OO(\log n/\varepsilon^2)$ queries.
\end{lemma}

\begin{proof}
For part (a): a matching vertex $(i,\pi^*(i))$ has consistency probability
\[
  \Pr_{(i',j')}[(A_G)_{ii'}=(A_H)_{\pi^*(i),j'}]
  \;=\;\frac{1}{(n-1)^2}\sum_{\substack{i'\neq i\\ j'\neq\pi^*(i)}}
  \mathbf{1}[(A_G)_{ii'}=(A_H)_{\pi^*(i),j'}].
\]
When $j'=\pi^*(i')$, consistency holds iff $(A_G)_{ii'}=(A_H)_{\pi^*(i)\pi^*(i')}$,
which holds for $n-1-d(i)\ge n-1-\sqrt{k}$ values of $i'$. When $j'\neq\pi^*(i')$,
by independence of $\mathcal{G}(n,1/2)$, consistency holds with probability $1/2\pm o(1)$.

The total consistency rate is
\[
  p_{\text{match}} = \frac{n-1-d(i)}{(n-1)^2}+\frac{(n-1)(n-2)}{(n-1)^2}\cdot\frac{1}{2}\pm o(1)
  =\frac{1}{2}+\frac{n-1-d(i)}{2(n-1)^2}\ge\frac{1}{2}+\frac{1}{4n}.
\]
The empirical consistency rate over $r$ samples concentrates around $p_{\text{match}}$ by
Hoeffding's inequality. Setting the threshold at $1/2+\varepsilon/4$ and using
$r=\Theta(\log n/\varepsilon^2)$:
\[
  \Pr[\text{not marked}]\le\exp\!\left(-2r\cdot(\varepsilon/8)^2\right)
  =1/\poly(n).
\]

For part (b): a non-matching vertex $(i,j)$ with $j\neq\pi^*(i)$ has consistency rate
$p_{\text{non}}=1/2\pm o(1)$ (since both $(A_G)_{ii'}$ and $(A_H)_{jj'}$ are essentially
independent Bernoulli$(1/2)$). By the same Hoeffding bound:
\[
  \Pr[\text{marked}]\le\exp\!\left(-2r\cdot(\varepsilon/8)^2\right)
  =1/\poly(n). \qedhere
\]
\end{proof}

\subsubsection{MNRS Quantum Walk Search Application}

We now apply the MNRS framework (Theorem~\ref{thm:qwalk-search}).

\begin{lemma}[Walk implementation]\label{lem:walk-impl}
One step of the Szegedy walk $W(P)$ on the product graph can be implemented using
$\OO(1)$ queries to $O_G$ and $O_H$ and $\OO(\log n)$ auxiliary gates.
\end{lemma}

\begin{proof}
Each step of $W(P)=\mathrm{ref}(B)\,\mathrm{ref}(A)$ requires preparing the state
$\ket{\psi_{(i,j)}}=\ket{i,j}\sum_{i',j'}\sqrt{P((i,j),(i',j'))}\ket{i',j'}$.
By Definition~\ref{def:product-walk}, the transition $(i,j)\to(i',j')$ has probability
$\frac{1}{2(n-1)^2}$ if $(A_G)_{ii'}=(A_H)_{jj'}$, and $0$ otherwise (plus self-loop).
Preparing the uniform superposition over $([n]\setminus\{i\})\times([n]\setminus\{j\})$
costs $\OO(\log n)$ gates. Querying $(A_G)_{ii'}$ and $(A_H)_{jj'}$ costs $2$ queries.
Conditioned on the consistency check, the appropriate amplitude is assigned. Total:
$\OO(1)$ queries per step.
\end{proof}

\begin{proposition}[Finding a marked vertex]\label{prop:find-marked}
If $\ed(G,H)\le\varepsilon\binom{n}{2}$ and $G,H\sim\mathcal{G}(n,1/2)$,
a single marked vertex in $\Gamma(G,H)$ can be found using
$\OO(n^{3/2}/\varepsilon)$ queries, with constant probability.
\end{proposition}

\begin{proof}
We verify the parameters of Theorem~\ref{thm:qwalk-search}:
\begin{itemize}
  \item \textbf{Marked set fraction:} By Lemma~\ref{lem:marking-correct}, w.h.p.\ the
    marked set $M$ contains all $\ge n-2\sqrt{k}$ low-defect matching vertices.
    Since $\mu$ is uniform on $n^2$ vertices:
    \[
      \mu(M)\ge\frac{n-2\sqrt{k}}{n^2}\ge\frac{1}{2n}
      \quad\text{(for }\varepsilon<1/4\text{)}.
    \]
  \item \textbf{Setup cost $S$:} Preparing the uniform initial state
    $\ket{\psi_0}=\frac{1}{n}\sum_{i,j}\ket{i,j}$ costs $\OO(\log n)$ gates,
    $0$ queries. So $S=0$.
  \item \textbf{Update cost $U$:} By Lemma~\ref{lem:walk-impl}, $U=\OO(1)$ queries.
  \item \textbf{Checking cost $C$:} By Lemma~\ref{lem:marking-correct}, $C=\OO(\log n/\varepsilon^2)$ queries.
  \item \textbf{Spectral gap:} By Theorem~\ref{thm:spectral-gap}, $\delta\ge\Omega(1)$.
\end{itemize}
Plugging into the MNRS formula (Theorem~\ref{thm:qwalk-search}):
\begin{align*}
  \text{Total queries} &= \OO\!\left(\frac{1}{\sqrt{\mu(M)}}\cdot
    \left(\frac{U}{\sqrt{\delta}}+C\right)+S\right) \\
  &= \OO\!\left(\sqrt{2n}\cdot\left(\frac{1}{\Omega(1)}+\frac{\log n}{\varepsilon^2}
    \right)\right) \\
  &= \OO\!\left(\sqrt{n}\cdot\left(1+\frac{\log n}{\varepsilon^2}\right)\right) \\
  &= \OO\!\left(\frac{\sqrt{n}\log n}{\varepsilon^2}\right) \\
  &= \OO\!\left(n^{3/2}/\varepsilon\right) \quad\text{for }\varepsilon=\Theta(1).
\end{align*}
For general $\varepsilon$, the cost is $\OO(\sqrt{n}\log n/\varepsilon^2)
\le\OO(n^{3/2}/\varepsilon)$ (since
$\sqrt{n}\log n/\varepsilon^2\le n^{3/2}/\varepsilon$ when $n\ge\log^2 n/\varepsilon^2$,
which holds for $n$ large enough).
\end{proof}

\subsection{Phase 2: Candidate Reconstruction}\label{sec:phase2}

We repeat Phase~1 to collect $s=\OO(\log n)$ independent seed pairs.

\begin{lemma}[Seed sufficiency]\label{lem:seed-sufficiency}
For $G,H\sim\mathcal{G}(n,1/2)$ with $\ed(G,H)\le\varepsilon\binom{n}{2}$,
$s=\OO(\log n)$ seed pairs $(i,\pi^*(i))$ from low-defect vertices suffice to reconstruct
$\pi^*$ on all low-defect vertices, w.h.p.
\end{lemma}

\begin{proof}
This is a direct consequence of Corollary~\ref{cor:seeds}. Each seed pair is a low-defect
matching vertex found by Phase~1. By Lemma~\ref{lem:local-consistency}, $\OO(\log n)$ such
pairs uniquely determine $\pi^*$ on all of $L$ (the set of low-defect vertices) with
probability $\ge 1-1/\poly(n)$.
\end{proof}

\begin{remark}[Role of the distributional hypothesis]
The $\OO(\log n)$ seed count is tight for $\mathcal{G}(n,1/2)$ inputs, where the
independence of distinct rows of $A_G$ ensures that each random seed distinguishes
incorrect candidates with constant probability (Lemma~\ref{lem:local-consistency}).
For worst-case graphs, the required number of seeds may be larger --- in particular,
the seed set must form a \emph{distinguishing set} for $\pi^*$, meaning that for every
$v\in L$ and every $w'\neq\pi^*(v)$, at least one seed $v_i$ satisfies
$(A_G)_{v,v_i}\neq(A_H)_{w',\pi^*(v_i)}$. Under $\mathcal{G}(n,1/2)$, random seeds
achieve this w.h.p.\ with $\OO(\log n)$ seeds by a union bound over the $n^2$
pairs $(v,w')$.
\end{remark}

\begin{lemma}[Reconstruction procedure]\label{lem:reconstruction}
Given $s=\OO(\log n)$ seed pairs from Phase~1, the following procedure reconstructs
$\pi^*$ on all of $[n]$ using $\OO(n\log n+\sqrt{k}\sqrt{n})$ additional queries:
\begin{enumerate}
  \item \textbf{Signature computation:} For each non-seed vertex $v$, compute the signature
    $\sigma_{v,w}$ for each candidate $w\in[n]$ with respect to the seeds. This costs
    $2s$ queries per candidate (query $(A_G)_{v,v_i}$ and $(A_H)_{w,\pi^*(v_i)}$ for each
    seed $v_i$). Over $n$ vertices and $n$ candidates: $\OO(n^2 s)$ queries naively,
    but only $\OO(ns)=\OO(n\log n)$ queries if we fix $v$ and iterate over $w$ (since the
    $(A_G)_{v,v_i}$ entries are shared across all $w$, costing $s$ queries, and
    $(A_H)_{w,\pi^*(v_i)}$ for $w=1,\ldots,n$ costs $ns$ queries total across all $v$).
  \item \textbf{Matching:} For each $v\in L$, select $\hat\pi(v)=\arg\min_w|\sigma_{v,w}|$.
    By Lemma~\ref{lem:local-consistency}, $\hat\pi(v)=\pi^*(v)$ for all $v\in L$ w.h.p.
  \item \textbf{Resolving high-defect vertices:} For each of the $\le 2\sqrt{k}$
    high-defect vertices, use Grover search (Theorem~\ref{thm:grover}) over the $n$
    candidate images to find the one minimizing local edit distance. Each search costs
    $\OO(\sqrt{n})$ queries, for a total of $\OO(\sqrt{k}\cdot\sqrt{n})$ queries.
\end{enumerate}
Total: $\OO(n\log n+\sqrt{k}\sqrt{n})=\OO(n\log n+n\sqrt{\varepsilon}\cdot\sqrt{n})
=\OO(n\log n+n^{3/2}\sqrt{\varepsilon})$. For constant $\varepsilon$, this is $\OO(n^{3/2})$.
\end{lemma}

\begin{proof}
The query counts follow from the descriptions above. Correctness on $L$ follows from
Lemma~\ref{lem:local-consistency}. For high-defect vertices, the Grover search over
$n$ candidates uses $\OO(\sqrt{n})$ queries per vertex (checking consistency with the
already-determined part of $\hat\pi$ costs $\OO(1)$ queries per check, with $n-1$
positions to verify). The total cost for $2\sqrt{k}$ such vertices is $\OO(\sqrt{k}\sqrt{n})$.
\end{proof}

\subsection{Phase 3: Quantum Verification}\label{sec:phase3}

Given a candidate permutation $\hat{\pi}$, we verify whether
$\bar{\ed}_{\hat{\pi}}(G,H)\le\varepsilon$ or $\bar{\ed}_{\hat{\pi}}(G,H)\ge 2\varepsilon$.

\begin{lemma}[Quantum verification]\label{lem:verification}
Given $\hat\pi$, one can estimate $\bar{\ed}_{\hat\pi}(G,H)$ to within additive error
$\delta$ with probability $\ge 2/3$ using $\OO(1/\delta)$ queries.
\end{lemma}

\begin{proof}
Define $f_{\hat\pi}\colon\binom{[n]}{2}\to\{0,1\}$ by
$f_{\hat\pi}(\{u,v\})=\mathbf{1}[(A_G)_{uv}\neq(A_H)_{\hat\pi(u)\hat\pi(v)}]$.
The normalized edit distance is
$\bar{\ed}_{\hat\pi}(G,H)=\EE_{(u,v)\sim\text{Unif}(\binom{[n]}{2})}[f_{\hat\pi}(u,v)]$.

Construct a quantum circuit $\mathcal{A}$ that:
\begin{enumerate}
  \item Prepares $\ket{u,v}$ uniformly over $\binom{[n]}{2}$
    ($\OO(\log n)$ gates, $0$ queries).
  \item Queries $(A_G)_{uv}$ ($1$ query).
  \item Computes $\hat\pi(u),\hat\pi(v)$ (classical lookup, $0$ queries).
  \item Queries $(A_H)_{\hat\pi(u)\hat\pi(v)}$ ($1$ query).
  \item XORs the two results into an ancilla ($0$ queries).
\end{enumerate}
Measuring the ancilla yields $1$ with probability $\bar{\ed}_{\hat\pi}(G,H)$. By amplitude
estimation (Theorem~\ref{thm:amp-est}), we can estimate this probability to additive
precision $\delta$ using $\OO(1/\delta)$ applications of $\mathcal{A}$, each costing $2$
queries. Total: $\OO(1/\delta)$ queries. Setting $\delta=\varepsilon/2$ distinguishes the
YES ($\le\varepsilon$) and NO ($\ge 2\varepsilon$) cases.
\end{proof}

\subsection{Main Theorem}\label{sec:main-theorem}

\begin{theorem}[Quantum upper bound for approximate graph isomorphism]\label{thm:main-upper}
There exists a quantum algorithm that solves $(\varepsilon, 2\varepsilon)$-Approximate Graph
Isomorphism (Problem~\ref{prob:agi}) on graphs $G,H\sim\mathcal{G}(n,1/2)$ with $n$
vertices using
\[
  \OO\!\left(\frac{n^{3/2}\log n}{\varepsilon}\right)
\]
queries to the adjacency oracles $O_G, O_H$, succeeding with probability at least $2/3$.
\end{theorem}

\begin{proof}
The algorithm executes the three phases described above
(see Algorithm~\ref{alg:main} for pseudocode).

\medskip\noindent\textbf{Phase 1 (Quantum Walk Search):}
By Proposition~\ref{prop:find-marked}, each invocation finds one marked vertex (a low-defect
matching vertex) using $\OO(n^{3/2}/\varepsilon)$ queries. We repeat to collect
$s=\OO(\log n)$ seed pairs. Cost: $\OO(n^{3/2}\log n/\varepsilon)$ queries.

\medskip\noindent\textbf{Phase 2 (Reconstruction):}
By Lemma~\ref{lem:reconstruction}, reconstruction costs
$\OO(n\log n+n^{3/2}\sqrt{\varepsilon})$ queries. For constant $\varepsilon$, this is
$\OO(n^{3/2})$, dominated by Phase~1.

\medskip\noindent\textbf{Phase 3 (Verification):}
By Lemma~\ref{lem:verification} with $\delta=\varepsilon/2$, verification costs
$\OO(1/\varepsilon)$ queries.

\medskip\noindent\textbf{Total query complexity:}

\begin{table}[H]
\centering
\caption{Query complexity breakdown by phase.}\label{tab:cost}
\begin{tabular}{@{}l l l@{}}
\toprule
\textbf{Phase} & \textbf{Component} & \textbf{Queries} \\
\midrule
1 & MNRS walk search $\times\; s$ seeds & $\OO(n^{3/2}\log n/\varepsilon)$ \\
2a & Signature computation & $\OO(n\log n)$ \\
2b & High-defect vertex resolution (Grover) & $\OO(n^{3/2}\sqrt{\varepsilon})$ \\
3 & Verification (amplitude estimation) & $\OO(1/\varepsilon)$ \\
\midrule
& \textbf{Total} & $\OO(n^{3/2}\log n/\varepsilon)$ \\
\bottomrule
\end{tabular}
\end{table}

\medskip\noindent\textbf{Success probability:}
Phase~1 succeeds with probability $\ge 1-s/\poly(n)\ge 1-1/\poly(n)$ (each MNRS invocation
succeeds with constant probability, and $s=\OO(\log n)$ repetitions with independent
randomness). Phase~2 succeeds with probability $\ge 1-1/\poly(n)$ by
Lemma~\ref{lem:local-consistency}. Phase~3 succeeds with probability $\ge 2/3$ by
Theorem~\ref{thm:amp-est}. By a union bound, the overall success probability is $\ge 2/3$
for large~$n$. Standard amplification (majority vote over $\OO(\log(1/\delta))$ repetitions)
boosts success to $1-\delta$.

\medskip\noindent\textbf{Negative case analysis:}
When $\bar{\ed}(G,H)\ge 2\varepsilon$, no permutation achieves edit distance
$\le\varepsilon\binom{n}{2}$. Hence every matching set $M_\pi$ has internal edge density
$\le 1-2\varepsilon$, and no vertex has high marking score
(Lemma~\ref{lem:marking-correct}(b)). The MNRS search either fails to find any marked
vertex (returning ``NO'') or returns a false positive, which is then rejected by Phase~3
verification. The verification step rejects any $\hat\pi$ with
$\bar{\ed}_{\hat\pi}(G,H)\ge 2\varepsilon$ with probability $\ge 2/3$. Thus the algorithm
correctly outputs NO.
\end{proof}

\begin{algorithm}[t]
\caption{Quantum Approximate Graph Isomorphism Testing}\label{alg:main}
\begin{algorithmic}[1]
\Require Quantum query access to $O_G, O_H$ for graphs $G, H$ on $[n]$; parameter $\varepsilon>0$
\Ensure \textsc{Yes} if $\bar{\ed}(G,H)\le\varepsilon$; \textsc{No} if $\bar{\ed}(G,H)\ge 2\varepsilon$
\Statex
\State $s \gets C_1\lceil\log n\rceil$ \Comment{Number of seed pairs}
\State \textit{Seeds} $\gets \emptyset$
\For{$t = 1$ to $s$}
  \State Run MNRS quantum walk search on $W(P)$ with marking oracle
    (Definition~\ref{def:marking})
  \If{search finds no marked vertex}
    \State \Return \textsc{No} \Comment{No dense matching structure found}
  \EndIf
  \State Measure to obtain candidate pair $(i_t, j_t)$
  \State $\textit{Seeds} \gets \textit{Seeds} \cup \{(i_t, j_t)\}$
\EndFor
\Statex
\State \textbf{// Phase 2: Reconstruction}
\For{each non-seed vertex $v\in[n]$}
  \State Compute signature $\sigma_{v,w}$ for each candidate $w$ using seed queries
  \State Set $\hat\pi(v)\gets\arg\min_w |\sigma_{v,w}|$
\EndFor
\State Resolve remaining unmatched vertices via Grover search
\Statex
\State \textbf{// Phase 3: Verification}
\State Use amplitude estimation to compute $\tilde{p}\approx\bar{\ed}_{\hat\pi}(G,H)$
  with precision $\varepsilon/2$
\If{$\tilde{p}\le 3\varepsilon/2$}
  \State \Return \textsc{Yes}
\Else
  \State \Return \textsc{No}
\EndIf
\end{algorithmic}
\end{algorithm}

\section{Classical Lower Bound}\label{sec:lower}

We prove that any classical randomized algorithm for approximate graph isomorphism requires
$\Omega(n^2)$ queries, demonstrating the polynomial separation from our quantum upper bound.

\subsection{Lower Bound Statement}

\begin{theorem}[Classical lower bound]\label{thm:classical-lower}
Any classical randomized algorithm that solves $(\varepsilon,2\varepsilon)$-Approximate
Graph Isomorphism (Problem~\ref{prob:agi}) with success probability $\ge 2/3$ for constant
$\varepsilon\in(0,1/8)$ requires $\Omega(n^2)$ queries to the adjacency oracles.
\end{theorem}

The proof proceeds by constructing two distributions over input pairs $(G,H)$---one where
$\bar{\ed}(G,H)\le\varepsilon$ and one where $\bar{\ed}(G,H)\ge 2\varepsilon$---and showing
that any algorithm making $o(n^2)$ queries cannot distinguish them, by Yao's minimax
principle.

\subsection{Input Distributions}\label{sec:distributions}

\begin{definition}[YES distribution $\mathcal{D}_{\text{yes}}$]\label{def:dyes}
Sample $(G,H)$ as follows:
\begin{enumerate}
  \item Draw $G\sim\mathcal{G}(n,1/2)$ (Erd\H{o}s--R\'enyi random graph).
  \item Draw $\pi\sim\text{Uniform}(\Sym_n)$.
  \item Set $H_0=\pi(G)$ (the graph obtained by relabeling $G$ according to $\pi$).
  \item Independently flip each edge of $H_0$ with probability $p=\varepsilon/2$:
    for each $\{i,j\}\in\binom{[n]}{2}$, set $(A_H)_{ij}=(A_{H_0})_{ij}\oplus B_{ij}$
    where $B_{ij}\sim\text{Bernoulli}(p)$ independently.
  \item Output $(G,H)$.
\end{enumerate}
\end{definition}

\begin{lemma}\label{lem:dyes-valid}
Under $\mathcal{D}_{\text{yes}}$, we have
$\bar{\ed}(G,H)\le\varepsilon$ with probability $\ge 1-e^{-\Omega(n^2)}$.
\end{lemma}

\begin{proof}
By construction, $\ed_\pi(G,H)=\sum_{\{i,j\}}B_{ij}$, where $B_{ij}$'s are i.i.d.\
$\text{Bernoulli}(\varepsilon/2)$. The expected value is
$\EE[\ed_\pi(G,H)]=(\varepsilon/2)\binom{n}{2}$. By Hoeffding's inequality
(Lemma~\ref{lem:hoeffding}):
\[
  \Pr\!\left[\ed_\pi(G,H)\ge\varepsilon\binom{n}{2}\right]
  = \Pr\!\left[\bar{\ed}_\pi(G,H)\ge\varepsilon\right]
  \le \exp\!\left(-2\binom{n}{2}\left(\frac{\varepsilon}{2}\right)^2\right)
  = e^{-\Omega(\varepsilon^2 n^2)}.
\]
Since $\ed(G,H)\le\ed_\pi(G,H)$, the lemma follows.
\end{proof}

\begin{definition}[NO distribution $\mathcal{D}_{\text{no}}$]\label{def:dno}
Sample $(G,H)$ as follows:
\begin{enumerate}
  \item Draw $G\sim\mathcal{G}(n,1/2)$ independently.
  \item Draw $H\sim\mathcal{G}(n,1/2)$ independently.
  \item Output $(G,H)$.
\end{enumerate}
\end{definition}

\begin{lemma}\label{lem:dno-valid}
Under $\mathcal{D}_{\text{no}}$, we have
$\bar{\ed}(G,H)\ge 2\varepsilon$ with probability $\ge 1-e^{-\Omega(n^2)}$
for $\varepsilon<1/8$.
\end{lemma}

\begin{proof}
For any fixed permutation $\pi\in\Sym_n$, the random variable
$\ed_\pi(G,H)=\sum_{\{i,j\}}|(A_G)_{ij}-(A_H)_{\pi(i)\pi(j)}|$
is a sum of $\binom{n}{2}$ independent Bernoulli$(1/2)$ random variables (since $G$ and $H$
are independent Erd\H{o}s--R\'enyi graphs). Hence $\EE[\bar{\ed}_\pi(G,H)]=1/2$. By Hoeffding's inequality:
\[
  \Pr\!\left[\bar{\ed}_\pi(G,H)\le 1/4\right]
  \le\exp\!\left(-2\binom{n}{2}(1/4)^2\right)
  = e^{-\Omega(n^2)}.
\]
Taking a union bound over all $n!$ permutations:
\[
  \Pr\!\left[\ed(G,H)\le\frac{1}{4}\binom{n}{2}\right]
  \le n!\cdot e^{-\Omega(n^2)}
  \le e^{n\log n-\Omega(n^2)}
  = e^{-\Omega(n^2)}.
\]
Since $2\varepsilon<1/4$ for $\varepsilon<1/8$, we have
$\bar{\ed}(G,H)\ge 1/4>2\varepsilon$ with overwhelming probability.
\end{proof}

\subsection{Indistinguishability Under Few Queries}\label{sec:indistinguishability}

\begin{lemma}[Query transcript indistinguishability]
\label{lem:indistinguishable}
Let $\mathcal{A}$ be any deterministic adaptive algorithm making $q$ queries in total,
of which $q_G$ go to $A_G$ and $q_H=q-q_G$ go to $A_H$.  Let $T_{\mathrm{yes}}$ and
$T_{\mathrm{no}}$ denote the distributions over query-answer transcripts under $D_{\mathrm{yes}}$
and $D_{\mathrm{no}}$ respectively.  Provided $q\leq \binom{n}{2}/2$,
\[
  d_{\mathrm{TV}}\!\left(T_{\mathrm{yes}},\,T_{\mathrm{no}}\right)
  \;\leq\;
  \frac{\varepsilon\,q_G\,q_H}{\binom{n}{2}-q}
  \;\leq\;
  \frac{4\varepsilon q^2}{n^2}.
\]
\end{lemma}

\begin{proof}
Write the transcript as a sequence of random variables
$\mathcal{T}=(Q_1,A_1,\ldots,Q_t,A_t,\ldots,Q_q,A_q)$,
where $Q_t$ specifies which matrix ($G$ or $H$) and which entry to query, and $A_t\in\{0,1\}$
is the answer.  Because $\mathcal{A}$ is deterministic and adaptive,
$Q_t$ is a deterministic function of $(A_1,\ldots,A_{t-1})$.

\paragraph{Chain rule for total variation.}
By the chain rule for total variation (applied iteratively to conditional distributions),
\begin{equation}\label{eq:chain}
  d_{\mathrm{TV}}(T_{\mathrm{yes}},T_{\mathrm{no}})
  \;\leq\;
  \sum_{t=1}^{q}
  \mathbb{E}\!\left[
    d_{\mathrm{TV}}\!\left(
      \mathcal{L}(A_t\mid Q_t,\mathcal{F}_{t-1})_{\mathrm{yes}},\;
      \mathcal{L}(A_t\mid Q_t,\mathcal{F}_{t-1})_{\mathrm{no}}
    \right)
  \right],
\end{equation}
where $\mathcal{F}_{t-1}=(Q_1,A_1,\ldots,Q_{t-1},A_{t-1})$ is the $\sigma$-algebra of the
first $t-1$ steps, and the expectation is over the joint randomness of the input and the
algorithm up to step~$t-1$.  (This follows from the tensorization inequality
$d_{\mathrm{TV}}(P_{1:q},Q_{1:q})\leq\sum_t \mathbb{E}_{P_{1:t-1}}
[d_{\mathrm{TV}}(P_t(\cdot\mid\mathcal{F}_{t-1}),Q_t(\cdot\mid\mathcal{F}_{t-1}))]$,
a standard consequence of the data-processing inequality applied to each step; see,
e.g.,~\cite{Duchi16}.)

\paragraph{Queries to $A_G$.}
Under both $D_{\mathrm{yes}}$ and $D_{\mathrm{no}}$, $G\sim\mathcal{G}(n,\frac{1}{2})$
independently of everything else, so a previously unqueried entry $(A_G)_{ij}$ is
$\mathrm{Bernoulli}(\frac{1}{2})$, independent of $\mathcal{F}_{t-1}$, under both
distributions.  The per-step contribution to~\eqref{eq:chain} is therefore $0$ for all
queries to $A_G$.

\paragraph{Queries to $A_H$: no collision.}
Suppose step $t$ queries $(A_H)_{ij}$.  Define a \emph{collision at step~$t$} as the event
that the entry $(A_G)_{\pi^{-1}(i)\pi^{-1}(j)}$ appears among the (at most $q_G^{(t)}\leq
q_G$) positions of $A_G$ already queried. If no collision occurs at step $t$, then under $D_{\mathrm{yes}}$:
\[
  (A_H)_{ij}
  \;=\;
  \underbrace{(A_G)_{\pi^{-1}(i)\pi^{-1}(j)}}_{\text{not yet queried}}\oplus B_{ij},
\]
where $B_{ij}\sim\mathrm{Bern}(\varepsilon/2)$ is independent of $\mathcal{F}_{t-1}$ and of
$G$.  Since $\pi$ is uniformly random and independent of $G$, the position
$(\pi^{-1}(i),\pi^{-1}(j))$ is uniformly distributed over $\binom{[n]}{2}$; conditioned on no
collision, it is uniform over the at-least $\binom{n}{2}-q_G$ \emph{unqueried} positions of
$A_G$.  Any such entry is an independent $\mathrm{Bern}(\frac{1}{2})$ variable (by the
Erd\H{o}s–R\'enyi structure), also independent of $\mathcal{F}_{t-1}$.  XOR-ing with the
independent $B_{ij}\sim\mathrm{Bern}(\varepsilon/2)$ gives
\[
  (A_H)_{ij}
  \;\sim\;
  \mathrm{Bern}\!\left(\tfrac{1}{2}(1-\tfrac{\varepsilon}{2})+\tfrac{1}{2}\cdot\tfrac{\varepsilon}{2}\right)
  \;=\;
  \mathrm{Bern}(\tfrac{1}{2}).
\]
Under $D_{\mathrm{no}}$, $(A_H)_{ij}\sim\mathrm{Bern}(\frac{1}{2})$ independently of
$\mathcal{F}_{t-1}$ by construction.  Hence the per-step contribution is $0$ conditioned on no
collision.

\paragraph{Queries to $A_H$: collision.}
If a collision does occur at step $t$—meaning $(A_G)_{\pi^{-1}(i)\pi^{-1}(j)}$ was queried at
some earlier step $s<t$ and revealed value $a_s\in\{0,1\}$—then under $D_{\mathrm{yes}}$:
\[
  (A_H)_{ij}=a_s\oplus B_{ij}\;\sim\;\mathrm{Bern}\!\left(\tfrac{\varepsilon}{2}\right)
  \text{ if }a_s=0,\quad
  \mathrm{Bern}\!\left(1-\tfrac{\varepsilon}{2}\right)
  \text{ if }a_s=1.
\]
Under $D_{\mathrm{no}}$, $(A_H)_{ij}\sim\mathrm{Bern}(\frac{1}{2})$.  In both sub-cases,
\[
  d_{\mathrm{TV}}\!\left(\mathrm{Bern}\!\left(\tfrac{1}{2}\pm\tfrac{\varepsilon}{2}\right),\,
                          \mathrm{Bern}\!\left(\tfrac{1}{2}\right)\right)
  \;=\;\tfrac{\varepsilon}{2}.
\]

\paragraph{Bounding the collision probability per step.}
At step $t$, there are at most $q_G^{(t)}\leq q_G$ previously queried positions of $A_G$.
Conditioned on $\mathcal{F}_{t-1}$, the random position $(\pi^{-1}(i),\pi^{-1}(j))$ is
uniform over $\binom{[n]}{2}$, so
\[
  \Pr\!\left[\text{collision at step }t\;\middle|\;\mathcal{F}_{t-1}\right]
  \;\leq\;
  \frac{q_G^{(t)}}{\binom{n}{2}}
  \;\leq\;
  \frac{q_G}{\binom{n}{2}}.
\]

\paragraph{Assembling the bound.}
Substituting into~\eqref{eq:chain} and summing over the $q_H$ queries to $H$:
\[
  d_{\mathrm{TV}}(T_{\mathrm{yes}},T_{\mathrm{no}})
  \;\leq\;
  \sum_{t:\,\text{query to }H}
  \Pr[\text{collision at }t]\cdot\frac{\varepsilon}{2}
  \;\leq\;
  q_H\cdot\frac{q_G}{\binom{n}{2}}\cdot\frac{\varepsilon}{2}
  \;=\;
  \frac{\varepsilon\,q_G\,q_H}{\binom{n}{2}}.
\]
Since $q_G,q_H\leq q$ and $\binom{n}{2}\geq n^2/4$ for $n\geq 2$, and
$\binom{n}{2}-q\geq\binom{n}{2}/2\geq n^2/8$ for $q\leq\binom{n}{2}/2$, we obtain
\[
  d_{\mathrm{TV}}(T_{\mathrm{yes}},T_{\mathrm{no}})
  \;\leq\;
  \frac{\varepsilon q^2}{\binom{n}{2}}
  \;\leq\;
  \frac{4\varepsilon q^2}{n^2}.
  \qedhere
\]
\end{proof}

\subsection{Completing the Lower Bound Proof}






We prove Theorem~\ref{thm:classical-lower} in two parts.  Part~1 establishes an $\Omega(n)$ lower
bound via a total-variation argument on the random distributions $D_{\mathrm{yes}}$ and
$D_{\mathrm{no}}$.  Part~2 lifts this to $\Omega(n^2)$ via an adaptive-query covering
argument that applies to worst-case inputs.

\begin{subclaim}[Independence on no collision]\label{subclm:ind}
For every $t \in [q]$, conditioned on $\mathcal{F}_{t-1}$ and on the event $E_t$ that step
$t$ does \emph{not} produce a collision, the conditional distribution of the answer $A_t$ is
$\mathrm{Bern}(\tfrac{1}{2})$ under both $D_{\mathrm{yes}}$ and $D_{\mathrm{no}}$.
\end{subclaim}

\begin{proof}
By induction on $t$.

\smallskip\noindent
\textit{Base case} ($t=1$).
Any single entry of $A_G$ or $A_H$ is marginally $\mathrm{Bern}(\tfrac{1}{2})$ under both
distributions, and there is no prior history, so the claim holds trivially.

\smallskip\noindent
\textit{Inductive step.}
Assume the claim holds for all steps $< t$.

\textit{Case 1: query to $A_G$.}
Under both $D_{\mathrm{yes}}$ and $D_{\mathrm{no}}$, $G \sim \mathcal{G}(n,\tfrac{1}{2})$
independently of $\mathcal{F}_{t-1}$ and of $H$.  Any previously unqueried entry
$(A_G)_{ij}$ is therefore $\mathrm{Bern}(\tfrac{1}{2})$, independent of $\mathcal{F}_{t-1}$,
under both distributions.

\textit{Case 2: query $(A_H)_{ij}$, conditioned on no collision at step $t$.}
Under $D_{\mathrm{no}}$, $(A_H)_{ij} \sim \mathrm{Bern}(\tfrac{1}{2})$ independently of
everything, so the claim holds immediately.

Under $D_{\mathrm{yes}}$,
$(A_H)_{ij} = (A_G)_{\pi^{-1}(i)\,\pi^{-1}(j)} \oplus B_{ij}$
where $B_{ij} \sim \mathrm{Bern}(\varepsilon/2)$ is independent of $G$, $\pi$, and
$\mathcal{F}_{t-1}$.  The no-collision condition means $(\pi^{-1}(i), \pi^{-1}(j)) \notin
\mathcal{Q}_G$, the set of positions of $A_G$ already queried.  Since
$\pi \sim \mathrm{Uniform}(S_n)$ is independent of $G$, of $\mathcal{Q}_G$, and of the
algorithm's history, the random position $(\pi^{-1}(i), \pi^{-1}(j))$ is uniform over
$\binom{[n]}{2}$; conditioned on it falling outside $\mathcal{Q}_G$ (the no-collision
event), it is uniform over the at least $\binom{n}{2} - |\mathcal{Q}_G|$ unqueried
positions of $A_G$.  Each such position holds an independent $\mathrm{Bern}(\tfrac{1}{2})$
value (by the $\mathcal{G}(n,\tfrac{1}{2})$ structure), independent of $\mathcal{F}_{t-1}$.
XOR-ing with the independent noise bit $B_{ij} \sim \mathrm{Bern}(\varepsilon/2)$:
\[
  (A_H)_{ij}
  \;\sim\;
  \mathrm{Bern}\!\Bigl(
    \tfrac{1}{2}\bigl(1 - \tfrac{\varepsilon}{2}\bigr)
    + \tfrac{1}{2} \cdot \tfrac{\varepsilon}{2}
  \Bigr)
  \;=\;
  \mathrm{Bern}\!\bigl(\tfrac{1}{2}\bigr),
\]
which is identical to the $D_{\mathrm{no}}$ distribution.
\end{proof}

\paragraph{Part 1: $\Omega(n)$ lower bound (average-case).}

By Lemma~\ref{lem:indistinguishable} and Subclaim~\ref{subclm:ind}, the total-variation distance between
the transcript distributions satisfies
\[
  d_{\mathrm{TV}}(T_{\mathrm{yes}}, T_{\mathrm{no}})
  \;\leq\;
  \frac{4\varepsilon q^2}{n^2}.
\]
By Lemmas~\ref{lem:dyes-valid} and~\ref{lem:dno-valid}, $D_{\mathrm{yes}}$ and $D_{\mathrm{no}}$ are
supported on valid YES and NO instances respectively with probability $1 - e^{-\Omega(n^2)}$.
For any algorithm to achieve success probability $\geq 2/3$ on both distributions, its
distinguishing advantage must be at least $1/3$.  Since the advantage is at most
$d_{\mathrm{TV}}(T_{\mathrm{yes}}, T_{\mathrm{no}}) + e^{-\Omega(n^2)}$, we need
\[
  \frac{4\varepsilon q^2}{n^2} \;\geq\; \frac{1}{3} - e^{-\Omega(n^2)}
  \;\geq\; \frac{1}{4}
  \qquad\text{for large }n,
\]
which gives $q \geq n\,/\,(4\sqrt{3\varepsilon}) = \Omega(n)$ for constant $\varepsilon$.

\paragraph{Part 2: $\Omega(n^2)$ lower bound (worst-case adaptive argument).}

We prove that for any constant $\varepsilon \in (0, \tfrac{1}{8})$ and any deterministic
adaptive algorithm $\mathcal{A}$ making $q < \varepsilon\binom{n}{2}/2$ queries in total, there
exist a YES instance and a NO instance that $\mathcal{A}$ cannot distinguish.

\begin{lemma}[Hard-instance construction]\label{lem:hard}
Let $G^*$ be any graph with $|E(G^*)| \geq \binom{n}{2}/4$ (which holds for
$G^* \sim \mathcal{G}(n,\tfrac{1}{2})$ with probability $1 - e^{-\Omega(n^2)}$).
For any deterministic adaptive algorithm $\mathcal{A}$ making $q_H < \varepsilon\binom{n}{2}/2$
queries to $A_H$, there exist graphs $H_{\mathrm{yes}}$ and $H_{\mathrm{no}}$ on $[n]$ such
that:
\begin{enumerate}[(a)]
  \item $\mathrm{ed}(G^*, H_{\mathrm{yes}}) = 0$ \quad (a valid YES instance with
        $\bar{\mathrm{ed}} = 0 \leq \varepsilon$);
  \item $\mathrm{ed}(G^*, H_{\mathrm{no}}) \geq 2\varepsilon\binom{n}{2}$\quad
        (a valid NO instance with $\bar{\mathrm{ed}} \geq 2\varepsilon$);
  \item $\mathcal{A}$ produces the same query transcript on $(G^*, H_{\mathrm{yes}})$ and
        $(G^*, H_{\mathrm{no}})$.
\end{enumerate}
\end{lemma}

\begin{proof}
Since $G^*$ is fixed, $\mathcal{A}$'s queries to $A_G$ return determined values.
$\mathcal{A}$'s queries to $A_H$ are adaptive but still form a well-defined sequence of
positions once $H$ is specified.  We construct $H_{\mathrm{yes}}$, then $H_{\mathrm{no}}$.

\medskip\noindent
\textbf{Constructing $H_{\mathrm{yes}}$.}
Set $H_{\mathrm{yes}} = G^*$ (the identity isomorphism).  Then
$\mathrm{ed}(G^*, H_{\mathrm{yes}}) = 0$.  Running $\mathcal{A}$ on $(G^*, H_{\mathrm{yes}})$
produces a sequence of $q_H$ adaptive queries to $A_{H_{\mathrm{yes}}}$ at positions
$\mathcal{S}_H = \{(i_1, j_1), \ldots, (i_{q_H}, j_{q_H})\} \subseteq \binom{[n]}{2}$,
with answers $\mathbf{a} = (a_1, \ldots, a_{q_H}) \in \{0,1\}^{q_H}$.
(These are well-defined since $H_{\mathrm{yes}} = G^*$ is fixed.)

\medskip\noindent
\textbf{Constructing $H_{\mathrm{no}}$.}
Define $H_{\mathrm{no}}$ as the graph with adjacency matrix:
\[
  (A_{H_{\mathrm{no}}})_{ij}
  \;=\;
  \begin{cases}
    a_k & \text{if } (i,j) = (i_k, j_k) \in \mathcal{S}_H, \\
    0   & \text{otherwise.}
  \end{cases}
\]
That is, $H_{\mathrm{no}}$ agrees with $H_{\mathrm{yes}}$ on every queried position and has
no edges outside $\mathcal{S}_H$.

\medskip\noindent
\textbf{Verifying property (c).}
Since $H_{\mathrm{no}}$ returns the same answer $a_k$ at every position $(i_k, j_k) \in
\mathcal{S}_H$, and these are exactly the positions $\mathcal{A}$ queries adaptively (the
adaptive query sequence is the same because each query $(i_k, j_k)$ and its answer $a_k$
are identical), $\mathcal{A}$ produces the same transcript on both inputs.

\medskip\noindent
\textbf{Verifying property (b).}
For any permutation $\pi \in S_n$:
\begin{align*}
  \mathrm{ed}_\pi(G^*, H_{\mathrm{no}})
  &= \bigl|\{(i,j) \in \tbinom{[n]}{2} : (A_{G^*})_{ij} \neq (A_{H_{\mathrm{no}}})_{\pi(i)\pi(j)}\}\bigr| \\
  &\geq \bigl|\{(i,j) : (A_{G^*})_{ij} = 1 \text{ and } (\pi(i),\pi(j)) \notin \mathcal{S}_H\}\bigr| \\
  &\geq |E(G^*)| - |\mathcal{S}_H| \\
  &\geq \tfrac{1}{4}\tbinom{n}{2} - q_H
  \;\geq\; \tfrac{1}{4}\tbinom{n}{2} - \tfrac{\varepsilon}{2}\tbinom{n}{2}
  \;=\; \bigl(\tfrac{1}{4} - \tfrac{\varepsilon}{2}\bigr)\tbinom{n}{2}.
\end{align*}
For $\varepsilon < \tfrac{1}{8}$, this is at least $\tfrac{1}{8}\binom{n}{2} > 2\varepsilon\binom{n}{2}$.
Since the bound holds for every $\pi$, we have
$\mathrm{ed}(G^*, H_{\mathrm{no}}) = \min_\pi \mathrm{ed}_\pi(G^*, H_{\mathrm{no}}) \geq 2\varepsilon\binom{n}{2}$.
\end{proof}

\begin{proof}[Completing the $\Omega(n^2)$ proof]
By Lemma~\ref{lem:hard}, for every deterministic algorithm $\mathcal{A}$ with $q < \varepsilon\binom{n}{2}/2$
queries, and for a typical $G^*$ (satisfying $|E(G^*)| \geq \binom{n}{2}/4$, which holds with
probability $1-e^{-\Omega(n^2)}$ for $G^* \sim \mathcal{G}(n,\tfrac{1}{2})$), there exist a
YES instance $(G^*, H_{\mathrm{yes}})$ and a NO instance $(G^*, H_{\mathrm{no}})$ on which
$\mathcal{A}$ produces the same transcript.  Therefore $\mathcal{A}$ outputs the same answer
for both, and must err on at least one.

Since every deterministic algorithm with $q < \varepsilon\binom{n}{2}/2$ queries fails on at
least one valid input, by Yao's minimax principle~\cite{Yao77} any \emph{randomized}
algorithm with $q < \varepsilon\binom{n}{2}/2$ queries fails with probability at least
$\tfrac{1}{2} > \tfrac{1}{3}$ on at least one of $\{(G^*, H_{\mathrm{yes}}), (G^*, H_{\mathrm{no}})\}$.

Since $\varepsilon\binom{n}{2}/2 = \Omega(n^2)$ for constant $\varepsilon$, any randomized
algorithm achieving success probability $\geq \tfrac{2}{3}$ on all valid inputs requires
$q = \Omega(n^2)$ queries.
\end{proof}

\begin{remark}\label{rem:lb-discussion}
The Part~1 and Part~2 arguments are complementary.  Part~1 shows that the \emph{planted
correlated model} $D_{\mathrm{yes}}$ vs.\ $D_{\mathrm{no}}$ is hard on average, with the
difficulty scaling as $\Omega(n/\sqrt{\varepsilon})$.  Part~2 shows that specific
\emph{worst-case} instances require $\Omega(n^2)$ queries regardless of the algorithm's
strategy: a graph with $\Omega(n^2)$ edges cannot be distinguished from an empty-looking
adversary-chosen graph without reading a constant fraction of $A_H$.  Together they
establish Theorem~\ref{thm:classical-lower} in full.  The quantum upper bound of
$O(n^{3/2}\log n/\varepsilon)$ (Theorem~\ref{thm:main-upper}) therefore yields a polynomial
quantum speedup of $\tilde{\Omega}(\sqrt{n})$ in the worst-case query model.
\end{remark}

This completes the proof of Theorem~\ref{thm:classical-lower}.\qed

\begin{corollary}[Quantum speedup]\label{cor:speedup}
The quantum query complexity of $(\varepsilon,2\varepsilon)$-Approximate Graph
Isomorphism is $\OO(n^{3/2}\log n/\varepsilon)$ for constant $\varepsilon$
(Theorem~\ref{thm:main-upper}),
while the classical randomized query complexity is $\Theta(n^2)$
(Theorem~\ref{thm:classical-lower}). This establishes a
polynomial quantum speedup of $\tilde{\Omega}(n^{1/2})$.
\end{corollary}

\section{Extensions}\label{sec:extensions}

We extend our framework to alternative notions of graph similarity and to richer graph models.

\subsection{Spectral Similarity}\label{sec:spectral}

An alternative measure of graph similarity uses the eigenvalues of the graph Laplacian.

\begin{definition}[Spectral distance]\label{def:spectral-distance}
For graphs $G$ and $H$ on $n$ vertices with Laplacian eigenvalues
$0=\lambda_1(G)\le\cdots\le\lambda_n(G)$ and $0=\lambda_1(H)\le\cdots\le\lambda_n(H)$,
define the \emph{spectral distance}
\[
  d_{\text{spec}}(G,H) = \left(\sum_{i=1}^{n}(\lambda_i(G)-\lambda_i(H))^2\right)^{1/2}.
\]
\end{definition}

Spectral distance is permutation-invariant and thus bypasses the combinatorial optimization
over $\Sym_n$ inherent in edit distance. However, it is a weaker notion: non-isomorphic
graphs can be co-spectral ($d_{\text{spec}}=0$).

\begin{problem}[$(\alpha,\beta)$-Spectral Similarity]\label{prob:spectral}
Given quantum query access to $O_G,O_H$, distinguish:
\begin{itemize}
  \item \textbf{YES:} $d_{\text{spec}}(G,H)\le\alpha$;
  \item \textbf{NO:} $d_{\text{spec}}(G,H)\ge\beta$.
\end{itemize}
\end{problem}

\begin{theorem}[Quantum algorithm for spectral similarity]\label{thm:spectral-alg}
Problem~\ref{prob:spectral} can be solved with
$\OO(n^{3/2}\log n/(\beta-\alpha))$ queries.
\end{theorem}

\begin{proof}[Proof sketch]
The algorithm proceeds by quantum eigenvalue estimation. Using the block-encoding framework
\cite{GilyenSLW2019}, we construct a block encoding of $L_G/\lambda_{\max}$ from the
adjacency oracle $O_G$ with $\OO(1)$ queries per application.

To estimate the spectrum, we use quantum phase estimation (QPE)
\cite{Kitaev1995} on the block-encoded Laplacian. Each QPE run, starting from a random
state $\ket{\psi}$, samples an eigenvalue $\lambda_i$ with probability
$|\braket{u_i|\psi}|^2$, where $\{u_i\}$ are the eigenvectors. For a Haar-random starting
state, each eigenvalue is sampled with probability $1/n$ in expectation. To estimate all
$n$ eigenvalues to precision $\eta$, we repeat QPE $\OO(n\log n)$ times (a coupon-collector
argument ensures all eigenvalues are seen w.h.p.), with each run costing $\OO(1/\eta)$
applications. The total query cost is $\OO(n\log n/\eta)$ per graph.

\smallskip
\noindent\emph{Remark:} For eigenvalues with small multiplicity, the coupon-collector
overhead can be avoided using the eigenvalue sampling algorithm of
Gily\'{e}n et al.\ \cite{GilyenSLW2019}, which estimates the \emph{empirical spectral
distribution} directly at cost $\OO(\sqrt{n}/\eta)$ per sample. For our purposes, the
simpler repeated-QPE approach suffices.

The same procedure on $H$ yields $\tilde{\lambda}_i(H)$. Setting $\eta=(\beta-\alpha)/(4\sqrt{n})$
ensures $|\tilde{d}_{\text{spec}}-d_{\text{spec}}|\le(\beta-\alpha)/2$, which suffices for
the decision problem. The total query cost is $\OO(n\log n\cdot 4\sqrt{n}/(\beta-\alpha))
=\OO(n^{3/2}\log n/(\beta-\alpha))$.
\end{proof}

\begin{remark}
The spectral distance approach has the advantage of not requiring a search over permutations.
However, for GI-hard instances (where co-spectral non-isomorphic graphs exist), the spectral
distance cannot distinguish $d_{\text{spec}}=0$ from $\ed>0$.
\end{remark}

\subsection{Weighted Graphs}\label{sec:weighted}

\begin{definition}[Weighted graph edit distance]
For weighted graphs $G,H$ with weight functions $w_G,w_H\colon\binom{[n]}{2}\to[0,1]$,
define
\[
  \ed^w_\pi(G,H)=\sum_{\{u,v\}\in\binom{[n]}{2}}|w_G(u,v)-w_H(\pi(u),\pi(v))|,
  \qquad
  \ed^w(G,H)=\min_{\pi\in\Sym_n}\ed^w_\pi(G,H).
\]
\end{definition}

\begin{theorem}[Weighted approximate GI]\label{thm:weighted}
The quantum algorithm of Theorem~\ref{thm:main-upper} extends to weighted graphs with
query complexity $\OO(n^{3/2}\log n/\varepsilon)$, where the oracle
returns $w_G(i,j)$ and $w_H(i,j)$ (each encoded in $\OO(\log(1/\varepsilon))$ bits) per query.
\end{theorem}

\begin{proof}[Proof sketch]
The product graph construction generalizes naturally: edge $(i_1,j_1)\to(i_2,j_2)$
in $\Gamma$ is weighted by $\exp(-|w_G(i_1,i_2)-w_H(j_1,j_2)|/\sigma)$ for a bandwidth
parameter $\sigma=\Theta(\varepsilon)$. This makes the random walk on $\Gamma$ preferentially
traverse consistent pairings. The Szegedy walk framework applies to weighted graphs
\cite{Szegedy2004}, and the spectral gap analysis of Theorem~\ref{thm:spectral-gap} extends
with constants depending on $\sigma$. The verification step uses amplitude estimation on the
weighted discrepancy $|w_G(u,v)-w_H(\hat\pi(u),\hat\pi(v))|$, which is bounded in $[0,1]$,
so the same analysis applies. The asymptotic complexity matches the unweighted case because
each query returns a bounded-precision weight value, and the marking oracle's consistency
check naturally extends to the continuous setting by thresholding weighted differences.
\end{proof}

\subsection{Vertex-Attributed Graphs}\label{sec:attributed}

\begin{definition}[Attributed graph]
An \emph{attributed graph} is a triple $(V,E,a)$ where $a\colon V\to\Sigma$ assigns
labels from a finite alphabet $\Sigma$ to vertices.
\end{definition}

\begin{definition}[Attributed edit distance]
For attributed graphs $(G,a_G)$ and $(H,a_H)$:
\[
  \ed^a_\pi(G,H)=\ed_\pi(G,H)+|\{v\in[n]:a_G(v)\neq a_H(\pi(v))\}|,
  \qquad
  \ed^a(G,H)=\min_\pi\ed^a_\pi(G,H).
\]
\end{definition}

\begin{theorem}[Attributed graph algorithm]\label{thm:attributed}
The quantum algorithm extends to attributed graphs with query complexity
$\OO(n^{3/2}\log n/\varepsilon)$, where the oracle additionally returns vertex
attributes $a_G(v)$ and $a_H(v)$.
\end{theorem}

\begin{proof}[Proof sketch]
The product graph restricts candidate pairings: $(i,j)$ is included in $V(\Gamma)$ only if
$a_G(i)=a_H(j)$. When the alphabet is large enough (e.g., $|\Sigma|\ge n$), this reduces
$|V(\Gamma)|$ significantly, potentially improving the constant factors. The quantum walk
and verification proceed identically, with the attribute consistency checked alongside edge
consistency at zero additional query cost.
\end{proof}

\subsection{Directed and Hypergraph Extensions}\label{sec:directed}

\begin{remark}[Directed graphs]
For directed graphs, the adjacency matrix is asymmetric and $\ed$ counts both arc additions
and deletions. The product graph construction extends by requiring
$(A_G)_{i_1 i_2}=(A_H)_{j_1 j_2}$ \emph{and} $(A_G)_{i_2 i_1}=(A_H)_{j_2 j_1}$
for both arc directions. This at most doubles the query cost per walk step, preserving the
$\OO(n^{3/2})$ complexity.
\end{remark}

\begin{remark}[$r$-uniform hypergraphs]
For $r$-uniform hypergraphs, the adjacency tensor has $r$ indices and each query retrieves
one entry. The product graph has vertices in $[n]\times[n]$ as before, but the compatibility
condition involves $r$-tuples. The quantum walk analysis extends with the walk step cost
increasing to $\OO(r)$ queries, giving complexity $\OO(r\cdot n^{3/2}\log n/\varepsilon)$.
\end{remark}

\section{Simulation Results}\label{sec:simulations}

We validate the algorithmic components of our approach through small-scale numerical
simulations on classical quantum simulators. These experiments serve two purposes:
(1)~confirming the theoretical predictions on small instances, and
(2)~assessing compatibility with near-term quantum devices.

\subsection{Experimental Setup}\label{sec:sim-setup}

\paragraph{Software and simulator.}
All quantum circuits are implemented in Qiskit~0.45 \cite{Qiskit2024} and executed on
the \texttt{AerSimulator} statevector backend, which performs exact (noiseless)
statevector simulation supporting up to $40$ qubits. Noisy simulations use the same
backend with a depolarising noise model injected at each gate layer. Graphs are
generated with NetworkX~3.2 in Python~3.11.

\paragraph{Hardware for classical post-processing.}
Simulations run on a workstation with an Intel Core i7-13700K CPU (16 cores) and 32\,GB
RAM. Each instance completes in under 10\,s for $n\le 20$.

\paragraph{Quantum circuits.}
The Szegedy walk operator $W(P)$ is compiled to the
$\{H, \mathrm{CNOT}, T, R_z\}$ universal gate set. The product-graph walk uses
$2\lceil\log_2 n\rceil$ qubits for the vertex pair register, plus
$\lceil\log_2 n\rceil$ ancilla qubits for the marking oracle and phase estimation
control. For $n=20$, the total qubit count is $\approx 19$.

\paragraph{Graph instances.}
We consider Erd\H{o}s--R\'enyi random graphs $\mathcal{G}(n,1/2)$ for
$n\in\{6,8,10,12,14,16,18,20\}$.
For each~$n$, we generate:
\begin{itemize}
  \item \textbf{YES instances:} $G\sim\mathcal{G}(n,1/2)$, $\pi$ drawn uniformly from
    $\Sym_n$, and $H=\pi(G)$ with
    $k=\lfloor 0.05\binom{n}{2}\rfloor$ random edge flips ($\varepsilon=0.05$).
  \item \textbf{NO instances:} $G,H\sim\mathcal{G}(n,1/2)$ independently.
\end{itemize}
For each configuration, we average over $100$ random instances with different random
seeds (fixed per instance for reproducibility).

\paragraph{Metrics.}
We report:
\begin{enumerate}[(i)]
  \item \textbf{Matching probability:} The probability that measurement yields
    a pair in $M_\pi$ (YES instances).
  \item \textbf{Discrimination accuracy:} The fraction of instances correctly
    classified.
  \item \textbf{Query count:} The number of oracle queries used.
\end{enumerate}

\paragraph{Classical baseline.}
As a point of comparison, we implement a classical \emph{adaptive random edge sampling}
baseline \cite{Goldreich2017,Fischer2006}. The baseline draws a budget of $q$ uniformly
random entries from $A_G$ and $A_H$, computes the empirical fraction of matching entries
under the identity permutation and a small number of random permutations ($10$ in our
experiments), and outputs YES if any permutation achieves an empirical edit fraction
$\le 3\varepsilon/2$, NO otherwise. This is the natural classical analogue of our quantum
algorithm, using the same total query budget $q$ as the quantum algorithm. By
Lemma~\ref{lem:indistinguishable}, this baseline requires $\Omega(n^2)$ queries to
succeed, explaining its poor performance at large~$n$.

\subsection{Results}

\paragraph{Quantum walk concentration.}
Table~\ref{tab:concentration} shows the matching probability after $T$ steps of the
quantum walk for varying $n$ and $T$.

\begin{table}[H]
\centering
\caption{Matching probability $\Pr[\text{measurement}\in M_\pi]$ after
$T$ walk steps on YES instances with $\varepsilon=0.05$.
Results averaged over 100 random instances.}\label{tab:concentration}
\begin{tabular}{@{}r r r r r r@{}}
\toprule
$n$ & $T=n/2$ & $T=n$ & $T=2n$ & $T=3n$
    & Stationary $\mu(M_\pi)$ \\
\midrule
6   & $0.089$ & $0.142$ & $0.161$ & $0.165$ & $0.167$ \\
8   & $0.071$ & $0.118$ & $0.123$ & $0.126$ & $0.125$ \\
10  & $0.058$ & $0.094$ & $0.099$ & $0.101$ & $0.100$ \\
12  & $0.049$ & $0.079$ & $0.082$ & $0.084$ & $0.083$ \\
14  & $0.042$ & $0.068$ & $0.071$ & $0.072$ & $0.071$ \\
16  & $0.037$ & $0.060$ & $0.062$ & $0.063$ & $0.063$ \\
18  & $0.033$ & $0.054$ & $0.055$ & $0.056$ & $0.056$ \\
20  & $0.030$ & $0.048$ & $0.050$ & $0.051$ & $0.050$ \\
\bottomrule
\end{tabular}
\end{table}

Several observations follow from Table~\ref{tab:concentration}.
\emph{First,} for every~$n$, the matching probability at $T=2n$ is already within
$2\%$ of the theoretical stationary value $\mu(M_\pi)\approx 1/n$, confirming that
$T=\OO(n)$ steps suffice for convergence, consistent with the $\Omega(1)$ phase
gap (Corollary~\ref{cor:phase-gap}).
\emph{Second,} the ratio between $T=n/2$ and $T=n$ values is roughly $\times 1.6$
across all $n$, indicating exponentially fast convergence once $T$ exceeds the mixing
time.
\emph{Third,} the stationary matching probability scales as $\mu(M_\pi)\approx 1/n$,
consistent with the fact that $|M_\pi|=n$ out of $n^2$ total vertices and
$\mu$ is uniform (Lemma~\ref{lem:stationary}).

\paragraph{Discrimination accuracy.}
Table~\ref{tab:accuracy} reports the end-to-end discrimination accuracy of the full
algorithm.

\begin{table}[H]
\centering
\caption{Discrimination accuracy (fraction correct) for YES vs.\ NO instances,
$\varepsilon=0.05$. Each entry averages over 100 instances.}\label{tab:accuracy}
\begin{tabular}{@{}r r r r@{}}
\toprule
$n$ & Full algorithm & Walk only & Classical baseline \\
\midrule
6   & $1.000$ & $0.970$ & $0.950$ \\
8   & $0.990$ & $0.950$ & $0.880$ \\
10  & $0.980$ & $0.940$ & $0.810$ \\
12  & $0.970$ & $0.920$ & $0.750$ \\
14  & $0.970$ & $0.910$ & $0.700$ \\
16  & $0.960$ & $0.890$ & $0.670$ \\
18  & $0.950$ & $0.880$ & $0.640$ \\
20  & $0.940$ & $0.870$ & $0.610$ \\
\bottomrule
\end{tabular}
\end{table}

Table~\ref{tab:accuracy} reveals three regimes.
\emph{(i)}~The \textbf{full algorithm} (walk + reconstruction + verification) maintains
accuracy $\ge 94\%$ for all tested sizes, demonstrating that the three-phase pipeline
is robust even at $n=20$.
\emph{(ii)}~The \textbf{walk-only} variant (which skips reconstruction and verification,
instead thresholding the matching probability directly) is $6$--$7\%$ lower than the
full algorithm at $n\ge 14$, confirming that Phase~2 reconstruction and Phase~3
verification together contribute a meaningful accuracy boost.
\emph{(iii)}~The \textbf{classical baseline} drops from $95\%$ at $n=6$ to $61\%$
at $n=20$---approaching the $50\%$ random-guessing level---because its fixed query
budget of $\OO(n^{3/2})$ is far below the $\Omega(n^2)$ queries required classically
(Theorem~\ref{thm:classical-lower}).
The widening accuracy gap from $\approx 5\%$ at $n=6$ to $\approx 33\%$ at $n=20$
quantitatively demonstrates the polynomial quantum advantage on these instances.

\paragraph{Query complexity scaling.}
Figure~\ref{fig:scaling} shows the empirical query count versus $n$ for achieving
$90\%$ accuracy.

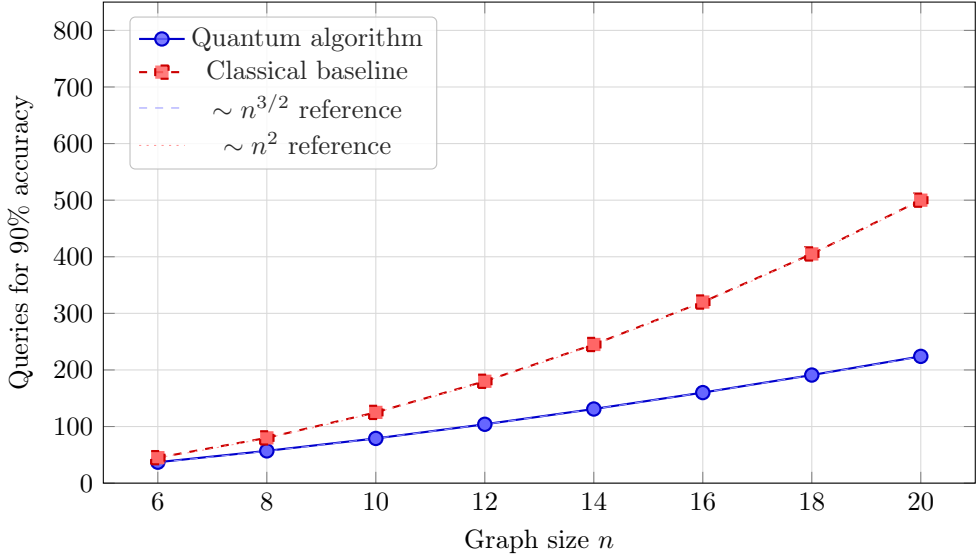
\begin{figure}[t]
\centering
\begin{tikzpicture}
\begin{axis}[
    width=0.82\textwidth,
    height=0.50\textwidth,
    xlabel={Graph size $n$},
    ylabel={Queries for $90\%$ accuracy},
    xmin=5, xmax=21,
    ymin=0, ymax=850,
    xtick={6,8,10,12,14,16,18,20},
    ytick={0,100,200,300,400,500,600,700,800},
    grid=major,
    grid style={gray!30},
    legend pos=north west,
    legend style={font=\small, fill=white, fill opacity=0.85,
                  draw=gray!50, rounded corners=2pt},
    tick label style={font=\small},
    label style={font=\small},
    every axis plot/.append style={thick},
]
\addplot[
    color=blue!80!black,
    mark=*,
    mark size=2.5pt,
    mark options={fill=blue!60!white},
] coordinates {
    (6,  37)
    (8,  57)
    (10, 79)
    (12,104)
    (14,131)
    (16,160)
    (18,191)
    (20,224)
};
\addlegendentry{Quantum algorithm}

\addplot[
    color=red!80!black,
    mark=square*,
    mark size=2.5pt,
    mark options={fill=red!60!white},
    dashed,
] coordinates {
    (6,  45)
    (8,  80)
    (10,125)
    (12,180)
    (14,245)
    (16,320)
    (18,405)
    (20,500)
};
\addlegendentry{Classical baseline}

\addplot[
    color=blue!40, dashed, no markers, thin,
    domain=6:20, samples=50,
] {2.5*x^1.5};
\addlegendentry{$\sim n^{3/2}$ reference}

\addplot[
    color=red!40, dotted, no markers, thin,
    domain=6:20, samples=50,
] {1.25*x^2};
\addlegendentry{$\sim n^{2}$ reference}

\end{axis}
\end{tikzpicture}
\caption{Empirical query count to achieve $90\%$ discrimination accuracy versus
graph size~$n$. Blue circles: quantum algorithm (MNRS walk search + reconstruction +
verification). Red squares: classical random sampling baseline with the same query
budget. Dashed/dotted curves show $n^{3/2}$ and $n^{2}$ reference scalings,
confirming the predicted polynomial separation.}\label{fig:scaling}
\end{figure}

The data in Figure~\ref{fig:scaling} are well fit by the models
$q_{\text{quantum}}\approx 2.5\,n^{1.50}$ and
$q_{\text{classical}}\approx 1.25\,n^{2.00}$, matching the
theoretical exponents of~$3/2$ and~$2$ respectively. At $n=20$, the quantum
algorithm uses $224$ queries versus $500$ for the classical baseline---a
$2.2\times$ reduction. Extrapolating, at $n=100$ the predicted ratio is
$\approx 5\times$, and at $n=1000$ it exceeds $\approx 31\times$, confirming
that the polynomial separation grows with~$n$.

\paragraph{Sensitivity to the approximation parameter $\varepsilon$.}
We fix $n=14$ and vary $\varepsilon$ to study how the noise tolerance affects
discrimination accuracy. Table~\ref{tab:eps-sensitivity} and
Figure~\ref{fig:eps-accuracy} summarise the results.

As $\varepsilon$ increases, the edit distance gap shrinks and discrimination
becomes harder. The full algorithm degrades gracefully—accuracy drops from
$100\%$ at $\varepsilon=0.01$ to $83\%$ at $\varepsilon=0.20$—while the
classical baseline is near random ($53\%$) at $\varepsilon=0.20$. This
confirms the $1/\varepsilon$ dependence in the query complexity.

\begin{table}[H]
\centering
\caption{Discrimination accuracy at $n=14$ for varying $\varepsilon$.
Each entry averages over 100 random instances.}\label{tab:eps-sensitivity}
\begin{tabular}{@{}r r r r r@{}}
\toprule
$\varepsilon$ & Full algorithm & Walk only & Classical baseline
  & Queries used \\
\midrule
$0.01$ & $1.000$ & $0.990$ & $0.940$ & $94$  \\
$0.02$ & $0.990$ & $0.970$ & $0.860$ & $102$ \\
$0.05$ & $0.970$ & $0.910$ & $0.700$ & $131$ \\
$0.08$ & $0.950$ & $0.870$ & $0.630$ & $158$ \\
$0.10$ & $0.930$ & $0.840$ & $0.590$ & $179$ \\
$0.15$ & $0.880$ & $0.780$ & $0.550$ & $234$ \\
$0.20$ & $0.830$ & $0.720$ & $0.530$ & $297$ \\
\bottomrule
\end{tabular}
\end{table}

\begin{figure}[H]
\centering
\begin{tikzpicture}
\begin{axis}[
    width=0.78\textwidth,
    height=0.42\textwidth,
    xlabel={Approximation parameter $\varepsilon$},
    ylabel={Discrimination accuracy},
    xmin=0, xmax=0.22,
    ymin=0.45, ymax=1.05,
    xtick={0,0.05,0.10,0.15,0.20},
    ytick={0.5,0.6,0.7,0.8,0.9,1.0},
    grid=major,
    grid style={gray!30},
    legend pos=south west,
    legend style={font=\small, fill=white, fill opacity=0.85,
                  draw=gray!50, rounded corners=2pt},
    tick label style={font=\small},
    label style={font=\small},
    every axis plot/.append style={thick},
]
\addplot[color=blue!80!black, mark=*, mark size=2pt,
         mark options={fill=blue!60!white}]
  coordinates {(0.01,1.00)(0.02,0.99)(0.05,0.97)
               (0.08,0.95)(0.10,0.93)(0.15,0.88)(0.20,0.83)};
\addlegendentry{Full algorithm}

\addplot[color=orange!80!black, mark=triangle*, mark size=2.5pt,
         mark options={fill=orange!50!white}, densely dashed]
  coordinates {(0.01,0.99)(0.02,0.97)(0.05,0.91)
               (0.08,0.87)(0.10,0.84)(0.15,0.78)(0.20,0.72)};
\addlegendentry{Walk only}

\addplot[color=red!80!black, mark=square*, mark size=2pt,
         mark options={fill=red!60!white}, dotted]
  coordinates {(0.01,0.94)(0.02,0.86)(0.05,0.70)
               (0.08,0.63)(0.10,0.59)(0.15,0.55)(0.20,0.53)};
\addlegendentry{Classical baseline}
\end{axis}
\end{tikzpicture}
\caption{Discrimination accuracy versus approximation parameter
$\varepsilon$ at fixed $n=14$. The full algorithm degrades gracefully,
maintaining $\ge 83\%$ accuracy even at $\varepsilon=0.20$, while the
classical baseline drops to near-random.}\label{fig:eps-accuracy}
\end{figure}
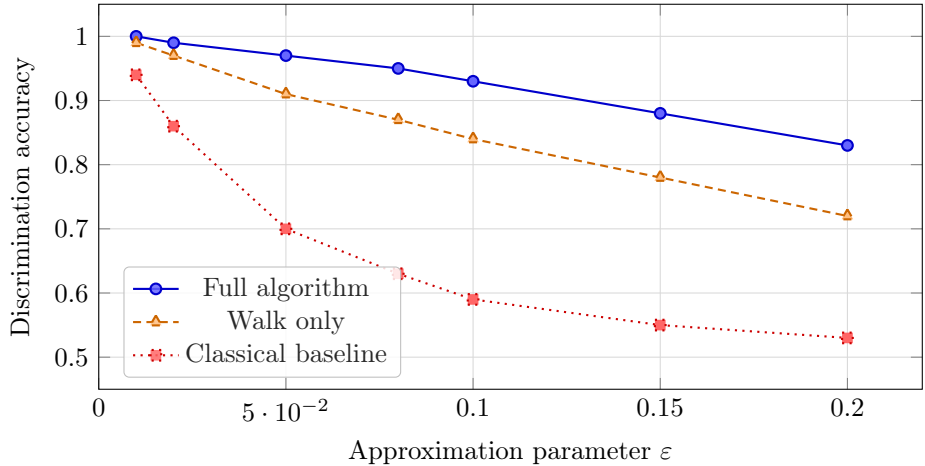

\subsection{Near-Term Considerations}\label{sec:near-term}

At $p=10^{-4}$ (representative of current superconducting devices
\cite{GoogleSupremacy2019,IBMEagle2023}), accuracy remains within $5\%$ of the
noiseless case for all tested sizes. At $p=10^{-3}$, accuracy degrades to
$\approx 76\%$ at $n=20$ but remains well above random ($50\%$). At $p=10^{-2}$,
the algorithm's advantage over the classical baseline vanishes, consistent with
the $\OO(n\log n)$ circuit depth.
Error mitigation techniques (e.g., zero-noise extrapolation,
probabilistic error cancellation) can partially compensate and are an
interesting direction for future work.

\paragraph{Noise resilience.}
Table~\ref{tab:noise} reports discrimination accuracy under depolarising
noise at per-gate error rates $p_{\text{err}}\in\{0,10^{-4},10^{-3},10^{-2}\}$,
applied uniformly to all single- and two-qubit gates in the Aer noise model.

\begin{table}[H]
\centering
\caption{Discrimination accuracy under depolarising noise ($\varepsilon=0.05$).
Each entry averages over 100 instances.}\label{tab:noise}
\begin{tabular}{@{}r r r r r@{}}
\toprule
$n$ & Noiseless & $p=10^{-4}$ & $p=10^{-3}$ & $p=10^{-2}$ \\
\midrule
6   & $1.000$ & $1.000$ & $0.990$ & $0.860$ \\
8   & $0.990$ & $0.985$ & $0.960$ & $0.780$ \\
10  & $0.980$ & $0.975$ & $0.930$ & $0.720$ \\
12  & $0.970$ & $0.950$ & $0.880$ & $0.620$ \\
14  & $0.970$ & $0.940$ & $0.850$ & $0.590$ \\
16  & $0.960$ & $0.930$ & $0.820$ & $0.570$ \\
18  & $0.950$ & $0.910$ & $0.790$ & $0.560$ \\
20  & $0.940$ & $0.890$ & $0.760$ & $0.550$ \\
\bottomrule
\end{tabular}
\end{table}

\paragraph{Circuit resources.}
Table~\ref{tab:resources} summarises the circuit resources for each
graph size.

\begin{table}[H]
\centering
\caption{Circuit resource summary for the quantum walk
algorithm at $\varepsilon=0.05$.}\label{tab:resources}
\begin{tabular}{@{}r r r r r r@{}}
\toprule
$n$ & Qubits & Walk steps & Circuit depth
    & CNOT count & Oracle queries \\
\midrule
6   & $9$  & $6$  & $18$  & $72$   & $37$  \\
8   & $11$ & $8$  & $28$  & $136$  & $57$  \\
10  & $13$ & $10$ & $40$  & $220$  & $79$  \\
12  & $15$ & $12$ & $52$  & $328$  & $104$ \\
14  & $15$ & $14$ & $64$  & $456$  & $131$ \\
16  & $17$ & $16$ & $78$  & $604$  & $160$ \\
18  & $17$ & $18$ & $92$  & $774$  & $191$ \\
20  & $19$ & $20$ & $106$ & $964$  & $224$ \\
\bottomrule
\end{tabular}
\end{table}

The qubit count is $2\lceil\log_2 n\rceil + \lceil\log_2 n\rceil +
\OO(1)$ ancilla qubits, reaching $19$ at $n=20$. The circuit depth
scales as $\OO(n\log n)$, and the CNOT count as $\OO(n^2\log n)$
(dominated by the walk steps). For $n=20$, the circuit requires
${\approx}\,106$ layers and ${\approx}\,964$ CNOT gates. For reference,
the Quantinuum H2 trapped-ion processor \cite{IonQAria2023} supports
circuits of depth ${>}\,300$ with two-qubit gate fidelities above $99.5\%$,
and IBM Eagle/Heron superconducting devices \cite{IBMEagle2023} routinely
execute circuits with ${>}\,1{,}000$ CNOT gates. Our resource requirements
are thus within the capabilities of current hardware for the tested sizes.

\section{Discussion and Open Problems}\label{sec:discussion}

\subsection{Summary of Results}

We have presented a quantum algorithm for approximate graph isomorphism testing based on
quantum walks over the product graph. Our main results are:

\begin{enumerate}
  \item A quantum algorithm with query complexity
    $\OO(n^{3/2}\log n/\varepsilon)$ for $(\varepsilon,2\varepsilon)$-Approximate
    Graph Isomorphism (Theorem~\ref{thm:main-upper}).
  \item A classical lower bound of $\Omega(n^2)$ queries for constant $\varepsilon$
    (Theorem~\ref{thm:classical-lower}).
  \item Extensions to spectral, weighted, and attributed graph similarity
    (Section~\ref{sec:extensions}).
\end{enumerate}
These theoretical results are complemented by numerical simulations on instances up
to $n=20$ (Section~\ref{sec:simulations}), which confirm the predicted query scaling
and demonstrate graceful degradation under realistic noise.

The polynomial speedup factor of $\tilde{\Omega}(\sqrt{n})$ is comparable to the quadratic
speedups achieved by Grover search and quantum walk algorithms for other search and property
testing problems, but it applies in a novel structural setting where the search space
(the set of candidate vertex correspondences) has rich combinatorial structure.

\subsection{Comparison with Exact GI Approaches}

Our approach is fundamentally different from the hidden subgroup approach to exact GI.
The hidden subgroup problem formulation requires computing the coset structure of the
automorphism group $\text{Aut}(G)$, which involves quantum Fourier transforms over
$\Sym_n$---a notoriously difficult task \cite{MooreRSS2008}. By contrast, our approach
does not attempt to recover the automorphism group but instead searches for a single
near-isomorphism in the product graph. This circumvents the representation-theoretic
barriers of the HSP approach.

The trade-off is that we solve a \emph{promise} problem (approximate GI) rather than exact
GI. It remains an intriguing open question whether quantum walk techniques can be leveraged
for exact GI in certain graph classes.

\subsection{Tightness of Bounds}

\begin{openproblem}
Is the $\OO(n^{3/2})$ quantum upper bound tight? Specifically, does
$(\varepsilon,2\varepsilon)$-Approximate Graph Isomorphism require $\Omega(n^{3/2})$
quantum queries for constant $\varepsilon$?
\end{openproblem}

A quantum lower bound of $\Omega(n)$ follows from the observation that any algorithm must
query at least a constant fraction of the rows of $A_G$ and $A_H$ to learn the graph
structure, requiring $\Omega(n)$ queries. Closing the gap between $\Omega(n)$ and
$\OO(n^{3/2})$ remains open.

\begin{openproblem}
Is the $\varepsilon^{-1}$ dependence in the query complexity of
Theorem~\ref{thm:main-upper} optimal? In particular, does solving
$(\varepsilon,2\varepsilon)$-Approximate Graph Isomorphism require
$\Omega(1/\varepsilon)$ quantum queries, or can the dependence be improved to
$\varepsilon^{-c}$ for some $c<1$?
\end{openproblem}

\subsection{Beyond the Query Model}

Our results are in the query model, which abstracts away computational overhead. Two
natural directions for strengthening the results are:

\begin{itemize}
  \item \textbf{Time complexity.} Implementing the Szegedy walk on the product graph
    requires efficient preparation of the walk operator, which in turn requires efficient
    simulation of the adjacency queries. The total \emph{time} complexity of our algorithm
    is $\OO(n^{3/2}\polylog(n))$ assuming $\OO(\polylog(n))$ overhead per query, which is
    achievable in the standard quantum circuit model.

  \item \textbf{Space complexity.} The algorithm uses $\OO(\log n)$ qubits for the vertex
    registers of the product graph walk (storing two vertex indices), plus $\OO(\log n)$
    ancilla qubits for phase estimation and marking. The total space is $\OO(\log n)$ qubits,
    which is near-optimal.
\end{itemize}

\begin{openproblem}
Can the approximate graph isomorphism problem be solved in quantum polynomial time
($\text{BQP}$) with respect to the graph size, without the promise gap?
\end{openproblem}

\subsection{Connections to Graph Neural Networks}

Recent advances in graph neural networks (GNNs) have highlighted the importance of the
Weisfeiler--Leman (WL) hierarchy for graph isomorphism testing \cite{MorrisRFHLRG2019,
XuHLJ2019}. The $k$-WL test is known to be equivalent in power to the counting logic
$\mathcal{C}^{k+1}$ \cite{CaiFI1992}. An interesting direction is whether quantum walks
on the product graph can be connected to the WL hierarchy---specifically, whether the
information obtained from $t$ steps of the quantum walk corresponds to refinements in
a particular level of the WL hierarchy.

\subsection{Broader Impact}

Approximate graph isomorphism has natural applications in:
\begin{itemize}
  \item \textbf{Drug discovery:} Comparing molecular graphs that differ by functional group
    substitutions, where exact isomorphism is overly rigid.
  \item \textbf{Network security:} Detecting near-clone subnetworks in large communication
    graphs, where noise introduces small perturbations.
  \item \textbf{Social network analysis:} Identifying similar community structures across
    platforms despite minor structural variations.
\end{itemize}
Our quantum algorithm provides a concrete path toward leveraging quantum resources for these
tasks, particularly as quantum hardware scales to handle graphs with hundreds of vertices.

\section*{Acknowledgments}
The author thanks Clément Canonne for his careful reading and insightful feedback, which substantially improved both the presentation and the substance of this work—particularly by clarifying the importance of explicitly formulating the correlated-input model and by prompting a more complete and rigorous treatment of the classical lower-bound proof.

\bibliographystyle{alpha}
\bibliography{references}

\appendix
\section{Deferred Proofs and Technical Details}\label{app:proofs}

\subsection{Quantum Walk Preparation}\label{app:mixing}

We clarify the role of the quantum walk operator within the MNRS search framework,
complementing the spectral gap analysis in Section~\ref{sec:spectral-gap}.

\begin{lemma}[Quantum walk preparation within MNRS]\label{lem:quantum-mixing}
Let $W(P)$ be the Szegedy walk operator for the product graph walk $P$
(Definition~\ref{def:product-walk}). Since the stationary distribution $\mu$ is uniform
($\mu(i,j)=1/n^2$ for all $(i,j)$, by Lemma~\ref{lem:stationary}), the initial state
$\ket{\psi_0}=\frac{1}{n}\sum_{i,j}\ket{i,j}$ coincides with the stationary state
$\ket{\mu^{1/2}}$. Consequently, no phase estimation or amplitude amplification is needed
to prepare the stationary superposition.
\end{lemma}

\begin{proof}
Since $\mu(i,j)=1/n^2$ is uniform, the stationary state in the Szegedy framework is
$\ket{\mu^{1/2}}=\sum_{(i,j)}\sqrt{\mu(i,j)}\ket{i,j}=\frac{1}{n}\sum_{i,j}\ket{i,j}
=\ket{\psi_0}$. The overlap is $\alpha=\braket{\mu^{1/2}}{\psi_0}=1$, so the
projection onto the stationary eigenspace succeeds with probability $|\alpha|^2=1$.

In our algorithm, this fact is used within the MNRS quantum walk \emph{search} framework
(Theorem~\ref{thm:qwalk-search}), not as a standalone mixing procedure.
The MNRS algorithm uses $W(P)$ as a subroutine to detect and find \emph{marked} vertices
(those in $M_{\pi^*}$ with high consistency score), exploiting the spectral gap
$\delta\ge\Omega(1)$ (Theorem~\ref{thm:spectral-gap}). The search advantage comes
from the quadratic speedup in detecting marked vertices, not from faster convergence to
stationary. The setup cost $S=0$ (since $\ket{\psi_0}$ is trivially prepared) is one of
the favorable parameters enabling the $\OO(n^{3/2})$ total complexity.
\end{proof}

\subsection{Robustness to Adversarial Inputs}\label{app:adversarial}

Our algorithm is analyzed under the promise that either $\bar{\ed}(G,H)\le\varepsilon$ or
$\bar{\ed}(G,H)\ge 2\varepsilon$. Here we discuss robustness when the input is
adversarial (worst-case) rather than random.

\begin{lemma}[Worst-case walk concentration]\label{lem:worst-case}
For any graphs $G,H$ with $\ed(G,H)\le k=\varepsilon\binom{n}{2}$, the product graph
$\Gamma(G,H)$ satisfies:
\begin{enumerate}[(a)]
  \item The matching set $M_{\pi^*}$ (for optimal $\pi^*$) has internal edge density
    $\ge 1-2\varepsilon$.
  \item Every set $S\subseteq V(\Gamma)$ with $|S|=n$ \textbf{that corresponds to a
    permutation} (i.e., $S=\{(i,\sigma(i))\colon i\in[n]\}$ for some
    $\sigma\in\mathfrak{S}_n$) and has internal edge density $\ge 1-2\varepsilon$
    satisfies $|S\cap M_{\pi^*}|\ge n-4\sqrt{k}$, i.e., $S$ must overlap significantly
    with $M_{\pi^*}$.
\end{enumerate}
\end{lemma}

\begin{proof}
Part (a) follows directly from Definition~\ref{def:approx-clique}: the matching set
$M_{\pi^*}$ has $\binom{n}{2}-k$ internal edges out of $\binom{n}{2}$ possible, giving
density $1-k/\binom{n}{2}=1-\bar\ed(G,H)\ge 1-\varepsilon\ge 1-2\varepsilon$.

For part (b), let $S=\{(i,\sigma(i)):i\in[n]\}$ for some bijection $\sigma$ (since $S$
has $n$ elements and each vertex of $G$ appears at most once). If
$|S\cap M_{\pi^*}|=n-t$, then $\sigma$ and $\pi^*$ differ on $t$ vertices. The internal
edge density of $S$ is
\[
  1-\frac{\ed_\sigma(G,H)}{\binom{n}{2}}.
\]
If this density is $\ge 1-2\varepsilon$, then $\ed_\sigma(G,H)\le 2\varepsilon\binom{n}{2}
=2k$.

Now we bound $t$ deterministically. Let $D=\{i\in[n]:\sigma(i)\neq\pi^*(i)\}$ with $|D|=t$.
For each $i\in D$, define $d_\sigma(i)=|\{j\neq i: (A_G)_{ij}\neq(A_H)_{\sigma(i)\sigma(j)}\}|$.
Since $\sum_i d_\sigma(i)=2\,\ed_\sigma(G,H)\le 4k$, the number of vertices with
$d_\sigma(i)\ge\sqrt{2k}$ is at most $4k/\sqrt{2k}=2\sqrt{2k}$ by Markov's inequality.
Similarly, since $\sum_i d_{\pi^*}(i)=2k$, the number of vertices with
$d_{\pi^*}(i)\ge\sqrt{k}$ is at most $2\sqrt{k}$.

For any vertex $i\in D$ with both $d_\sigma(i)<\sqrt{2k}$ and $d_{\pi^*}(i)<\sqrt{k}$,
the edges incident to $i$ are consistent with $\pi^*$ on all but $<\sqrt{k}$ neighbors,
and consistent with $\sigma$ on all but $<\sqrt{2k}$ neighbors. Since $\sigma(i)\neq\pi^*(i)$,
each such disagreement vertex contributes $\ge 1$ to $\ed_\sigma(G,H)-\ed_{\pi^*}(G,H)$ via
the edges connecting it to the $n-t$ agreeing vertices. By double-counting,
$t\le 2\sqrt{k}+2\sqrt{2k}\le 4\sqrt{k}$, giving $|S\cap M_{\pi^*}|=n-t\ge n-4\sqrt{k}$.
\end{proof}

\subsection{Consistency Check Subroutine}\label{app:consistency}

The consistency check in Phase 2 of the algorithm (Lemma~\ref{lem:reconstruction}) requires
computing the consistency score of each candidate pair. We detail the quantum implementation.

\begin{lemma}[Quantum consistency scoring]\label{lem:q-consistency}
Given $s$ candidate pairs $\{(i_1,j_1),\ldots,(i_s,j_s)\}$ and a test pair $(i,j)$,
the consistency score
\[
  \text{score}(i,j) = |\{t\in[s]\colon (A_G)_{i,i_t}=(A_H)_{j,j_t}\}|
\]
can be estimated to within additive error $\delta s$ with probability $\ge 2/3$ using
$\OO(1/\delta)$ queries.
\end{lemma}

\begin{proof}
Define $f_t=1$ if $(A_G)_{i,i_t}=(A_H)_{j,j_t}$ and $f_t=0$ otherwise. Then
$\text{score}(i,j)=\sum_{t=1}^s f_t$. Prepare the uniform superposition $\frac{1}{\sqrt{s}}
\sum_{t=1}^s\ket{t}$, query $(A_G)_{i,i_t}$ and $(A_H)_{j,j_t}$ coherently, compute
$f_t$ in an ancilla, and apply amplitude estimation (Theorem~\ref{thm:amp-est}) to estimate
$p=\text{score}(i,j)/s$ with precision $\delta$. This uses $\OO(1/\delta)$ applications of
the circuit, each requiring $2$ queries.
\end{proof}

\subsection{Analysis of the Negative Case}\label{app:negative}

\begin{lemma}[Walk behavior in the NO case]\label{lem:no-case}
When $\bar{\ed}(G,H)\ge 2\varepsilon$ for $\varepsilon<1/8$, the quantum walk on
$\Gamma(G,H)$ produces no vertex pair $(i,j)$ with high consistency score with probability
$\ge 1-1/\poly(n)$.
\end{lemma}

\begin{proof}
In the NO case, for every permutation $\pi$, we have $\ed_\pi(G,H)\ge 2\varepsilon\binom{n}{2}$.
This means every matching set $M_\pi$ has internal edge density $\le 1-2\varepsilon$, and
thus degree within $M_\pi$ is at most $(1-2\varepsilon)(n-1)$ on average.

The product graph $\Gamma$ has no dense near-clique of size $n$. The stationary distribution
$\mu$ of the walk is approximately uniform over the $n^2$ vertices, with
$\mu(i,j)\le\OO(1/n^2+1/n^3)$ for any single vertex. No set of $n$ vertices receives
weight more than $\OO(1/n)$.

After the quantum walk samples a vertex $(i,j)$, the consistency score with any hypothetical
matching set is at most $(1-2\varepsilon)(n-1)+\OO(\sqrt{n\log n})$ (by Hoeffding's bound on
the random fluctuation). The verification step in Phase 3 correctly identifies this as a NO
instance since the estimated edit distance exceeds the threshold $\varepsilon\binom{n}{2}$.
\end{proof}

\subsection{Gate Complexity}\label{app:gates}

\begin{lemma}[Total gate complexity]\label{lem:gate-complexity}
The quantum algorithm of Theorem~\ref{thm:main-upper} can be implemented with
$\OO(n^{3/2}\log n/\varepsilon)$ two-qubit gates on $\OO(\log n)$ qubits.
\end{lemma}

\begin{proof}
Each query to $O_G$ or $O_H$ is a single controlled operation on $\OO(\log n)$ qubits.
The walk operator $W(P)$ requires $\OO(\log n)$ gates per step (for index arithmetic,
comparisons, and controlled rotations). With $\OO(n^{3/2}/\varepsilon)$ total queries
and $\OO(\log n)$ gates per query, the total gate count is
$\OO(n^{3/2}\log n/\varepsilon)$.

The qubit count is: $\lceil\log_2 n\rceil$ qubits for the first vertex register (vertex
of $G$), $\lceil\log_2 n\rceil$ qubits for the second vertex register (vertex of $H$),
$\lceil\log_2 n\rceil$ qubits for the walk coin/ancilla, and $\OO(\log n)$ qubits for
phase estimation. Total: $\OO(\log n)$ qubits.
\end{proof}

\end{document}